\documentclass[autoref,thm-restate]{lipics-v2021}
\nolinenumbers
\pdfoutput=1 
\hideLIPIcs 
\usepackage{xfrac}

\newcommand{\N}{\mathbb{N}}

\newcommand{\W}{\mathrm{W}}
\newcommand{\WP}{\mathrm{W[P]}}
\newcommand{\WSAT}{\mathrm{W[SAT]}}
\newcommand{\A}{\mathrm{A}}
\newcommand{\C}{\mathcal{C}}
\renewcommand{\S}{\mathrm{S}}
\newcommand{\XP}{\mathrm{XP}}
\newcommand{\XSLP}{\mathrm{XSLP}}
\newcommand{\XNLP}{\mathrm{XNLP}}
\newcommand{\XALP}{\mathrm{XALP}}
\newcommand{\pl}{\mathrm{pl}}
\newcommand{\fpt}{\mathrm{fpt}}

\newcommand{\vc}{\mathrm{vc}}

\newcommand{\td}{\mathrm{td}}

\newcommand{\csp}{\textsc{BinCSP}}
\newcommand{\Oh}{{\mathcal{O}}}

\renewcommand{\leq}{\leqslant}
\renewcommand{\geq}{\geqslant}

\title{Parameterized Complexity of Binary CSP:\texorpdfstring{\\}{} Vertex Cover, Treedepth, and Related Parameters}

\author{Hans L. Bodlaender}{Department of Information and Computing Sciences, Utrecht University}{h.l.bodlaender@uu.nl}{ https://orcid.org/
0000-0002-9297-3330}{}

\author{Carla Groenland}{Faculty of Electrical Engineering, Mathematics and Computer Science, Technical University Delft}{c.e.groenland@tudelft.nl}{https://orcid.org/
0000-0002-9878-8750}{Supported by the Marie Skłodowska-Curie grant GRAPHCOSY (number 101063180).\flag[0.2\textwidth]{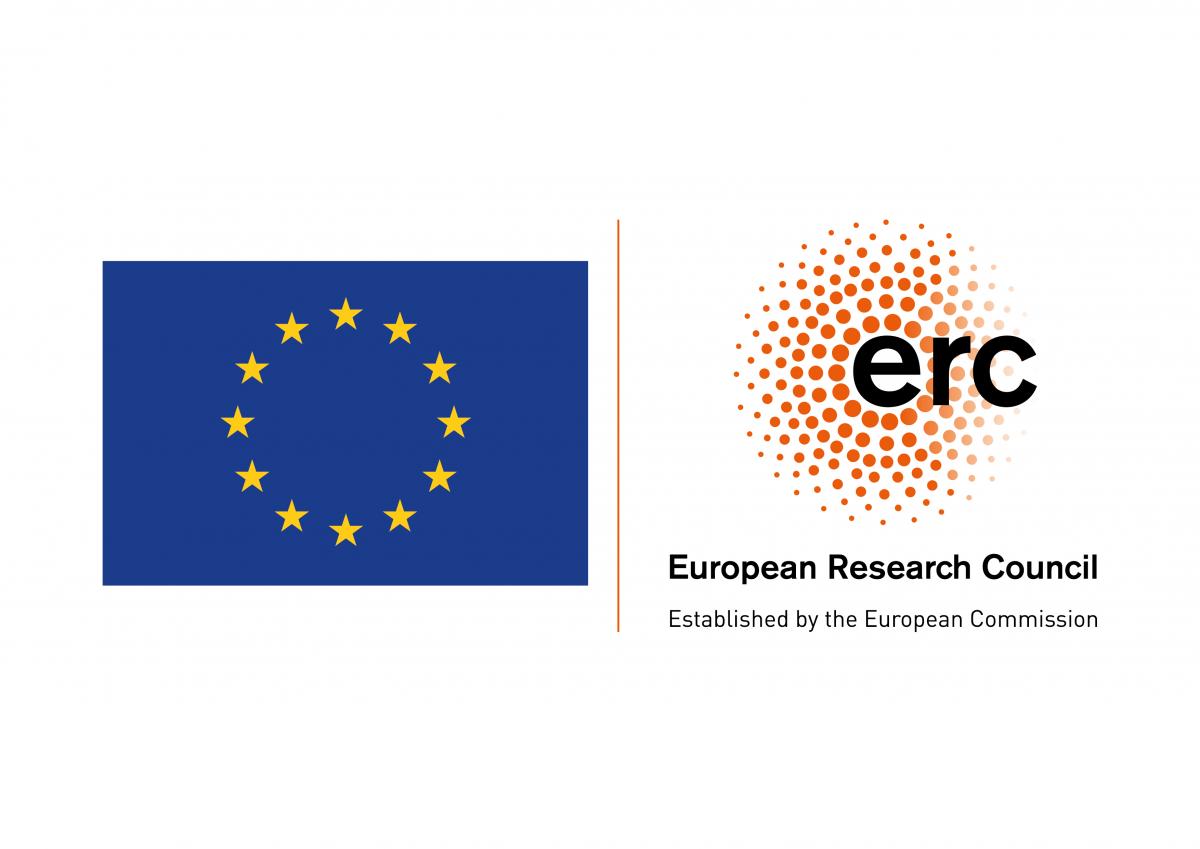}. This research was done when Carla Groenland was associated with Utrecht University.}

\author{Micha\l{} Pilipczuk}{Institute of Informatics, University of Warsaw}{michal.pilipczuk@mimuw.edu.pl}{https://orcid.org/
0000-0001-7891-1988}{This research is a part of the project BOBR that has received funding from the European Research Council (ERC)
under the European Union's Horizon 2020 research and innovation programme
(grant agreement no. 948057).}

\authorrunning{H. L. Bodlaender, C. Groenland and Mi. Pilipczuk} 
\titlerunning{Binary CSP: Vertex Cover, Treedepth, and Related Parameters}

\Copyright{Hans L. Bodlaender, Carla Groenland and Micha\l{} Pilipczuk} 



\relatedversion{An extended abstract of this paper 
appeared in the proceedings of ICALP 2023, 
at \href{https://doi.org/10.4230/LIPIcs.ICALP.2023.27}{DOI: 10.4230/LIPIcs.ICALP.2023.27}.} 

\date{\today}
\begin{document}

\maketitle
\begin{abstract}
We investigate the parameterized complexity of \textsc{Binary CSP} parameterized by the vertex cover number and the treedepth of the constraint graph, as well as by a selection of related modulator-based parameters. The main findings are as follows:
\begin{itemize}
    \item \textsc{Binary CSP} parameterized by the vertex cover number is $\mathrm{W}[3]$-complete. More generally, for every positive integer $d$, \textsc{Binary CSP} parameterized by the size of a modulator to a treedepth-$d$ graph is $\mathrm{W}[2d+1]$-complete. This provides a new family of natural problems that are complete for odd levels of the $\mathrm{W}$-hierarchy.
    \item We introduce a new complexity class $\mathrm{XSLP}$, defined so that \textsc{Binary CSP} parameterized by treedepth is complete for this class. We provide two equivalent characterizations of $\mathrm{XSLP}$: the first one relates $\mathrm{XSLP}$ to a model of an alternating Turing machine with certain restrictions on conondeterminism and space complexity, while the second one links $\mathrm{XSLP}$ to the problem of model-checking first-order logic with suitably restricted universal quantification. Interestingly, the proof of the machine characterization of $\mathrm{XSLP}$ uses the concept of {\em{universal trees}}, which are prominently featured in the recent work on parity games.
    \item We describe a new complexity hierarchy sandwiched between the $\mathrm{W}$-hierarchy and the $\mathrm{A}$-hierarchy: For every odd $t$, we introduce a parameterized complexity class $\mathrm{S}[t]$ with  $\mathrm{W}[t]\subseteq \mathrm{S}[t]\subseteq \mathrm{A}[t]$, defined using a parameter that interpolates between the vertex cover number and the treedepth. 
\end{itemize}
We expect that many of the studied classes will be useful in the future for pinpointing the complexity of various structural parameterizations of graph problems.
\end{abstract}

\section{Introduction}\label{sec:intro}
The \textsc{Binary Constraint Satisfaction Problem} ($\csp$, for short) is a fundamental problem defined as follows. We are given an undirected graph $G=(V,E)$, called the {\em{primal}} or the {\em{Gaifman graph}}, where $V$ is a set of variables, each with a prescribed domain of possible values. Further, each edge $uv$ of $G$ corresponds to a binary constraint that restricts the possible pairs of values that can be assigned to $u$ and $v$. The task is to decide whether every variable can be mapped to a value from its domain so that all the constraints are satisfied.

Due to immense modeling power, constraint satisfaction problems are of great importance in multiple applications, and the theoretical study of their complexity is a field on its own. In this work we are interested in parameterized algorithms for $\csp$, with a particular focus on structural parameters of the Gaifman graph. An example of such a result is a classic observation, usually attributed to Freuder~\cite{Freuder90}: using dynamic programming, $\csp$ can be solved in time $n^{k+\Oh(1)}$, where $n$ is the maximum size of a domain and $k$ is the treewidth of the Gaifman graph. In the language of parameterized complexity, this means that $\csp$ parameterized by treewidth is {\em{slice-wise polynomial}}, or in the complexity class $\XP$.

The class $\XP$ is very general and just placing $\csp$ parameterized by treewidth within $\XP$ does not provide much insight into the actual complexity of the problem. A more detailed study of the parameterizations of $\csp$ by pathwidth and by treewidth was recently performed by Bodlaender, Groenland, Nederlof, and Swennenhuis in~\cite{BodlaenderGNS21}, and by Bodlaender, Groenland, Jacob, Pilipczuk, and Pilipczuk in~\cite{BodlaenderGJJPP22}. In particular, as shown in~\cite{BodlaenderGNS21}, $\csp$ parameterized by pathwidth is complete for $\XNLP$: the class of all parameterized problems that can be solved by a nondeterministic Turing machine using $f(k)\log n$ space and $f(k)\cdot n^{\Oh(1)}$ time, where $k$ is the parameter and $f$ is a computable function. A ``tree variant'' of $\XNLP$, called $\XALP$, was studied in~\cite{BodlaenderGJJPP22}; it can be defined using the same model of a Turing machine, except that the machine additionally has access to a stack of unbounded size that can be manipulated by pushing and popping. As proved in~\cite{BodlaenderGJJPP22}, $\csp$ parameterized by treewidth is complete for $\XALP$. All in all, the recent works~\cite{BodlaenderCW22,BodlaenderGJJL22,BodlaenderGJJPP22,BodlaenderGNS21,ElberfeldST15} present a variety of problems on graphs with linear or tree-like structure that are complete for $\XNLP$ and $\XALP$, respectively. This is an evidence that $\XNLP$ and $\XALP$ capture certain fundamental varieties of computational problems: those amenable to linearly and tree-structured dynamic programming with  state space of slice-wise polynomial size.

The contemporary research in parameterized algorithms features many more structural parameters of graphs, besides treewidth and pathwidth. In this work we explore the complexity of $\csp$ parameterized by the following parameters of the Gaifman graph: (1)~the vertex cover number, (2) the treedepth, and (3) a selection of related modulator-based parameters lying between the vertex cover number and the treedepth. 

\subparagraph*{New completeness results for the $\W$-hierarchy.} 
The $\W$-hierarchy was introduced around thirty years ago
in the work by Downey and Fellows that founded the field of parameterized
algorithms and complexity. In this hierarchy, we have a collection of classes,
including $\W[1]\subseteq \W[2]\subseteq \ldots \subseteq \W[\mathrm{SAT}]\subseteq \W[\mathrm{P}]$; see~\cite{DowneyF99,DowneyF13,FlumGrohe} for an overview and for bibliographic references.
A large variety of problems are known to be complete (under fpt reductions) for $\W[1]$ and for $\W[2]$. However, for classes $\W[t]$ with
 $t\geq 3$, there is so far only a handful of examples
 of natural problems known to be complete~\cite{AbuKhzamFGLM22,BlasiusFLMS22,BlasiusFS22,CaselFGMS22,ChenZ06,HannulaSL21}. Our first contribution is to give new examples of complete problems for $\W[t]$ for all odd $t\geq 3$.

Our first example concerns $\csp$ parameterized by the {\em{vertex cover number}}: the minimum size of a vertex cover in the Gaifman graph.
\begin{theorem}
\label{thm:vertex_cover}
$\csp$ parameterized by the vertex cover number of the Gaifman graph is complete
for the class $\W[3]$.
\end{theorem}
It was known that $\csp$ parameterized by the vertex cover number is $\W[1]$-hard~\cite{FialaGK11,PapadimitriouY99}. The $\W[3]$-completeness is surprising, not only due to the small number of examples of natural $\W[3]$-complete problems, but also because many problems appear
to be fixed-parameter tractable or even have a kernel of polynomial size, when the vertex cover number is used as the parameter (e.g.,~\cite{FellowsLMRS08,FialaGK11,FominJP14,Ganian15}). 

For a graph $G$ and a graph class $\mathcal{C}$, a \emph{modulator} to $\mathcal{C}$ in $G$ is a set of vertices $W$ such that $G-W\in \mathcal{C}$. For instance, vertex covers are modulators to the class of edgeless graphs. A {\em{feedback vertex set}} is another type of a
modulator, now to graphs without cycles, i.e., to forests. The {\em{feedback vertex number}} of a graph $G$ is the minimum size of a feedback vertex set in $G$. We prove that the parameterization by the feedback vertex number yields a much harder problem.

\begin{theorem}
\label{thm:fvs}
$\csp$ parameterized by the feedback vertex number of the Gaifman graph is $\WSAT$-hard and in $\WP$.
\end{theorem}

Finally, with similar techniques, we obtain the following completeness results for $\W[t]$ for all odd $t\geq 3$. Here, {\em{treedepth}} is a structural parameter measuring the ``depth'' of a graph, we will expand on it later on.

\begin{theorem}
\label{thm:oddW}
For each integer $d\geq 1$,
$\csp$ is complete for $\W[2d+1]$ when parameterized by the minimum size of a modulator to a graph of treedepth at most $d$, and when parameterized by the minimum size of a modulator to a forest of depth at most $d$.
\end{theorem}

Interestingly, each increase of the depth of the trees by one corresponds to an increase in the W-hierarchy by two levels: this
is because one level of depth in the tree or forest  corresponds to a conjunction (looking at all children of a node) with a
disjunction (the choice of a value). 
Theorem~\ref{thm:oddW} can be seen as an interpolation between Theorems~\ref{thm:vertex_cover} and~\ref{thm:fvs}: by allowing the forest to have larger and larger depth, we obtain harder and harder parameterized problems. This yields a family of natural complete problems for the odd levels of the $\W$-hierarchy.

Theorem~\ref{thm:vertex_cover} is proved in Section~\ref{section:hardness}, and Theorems~\ref{thm:fvs} and~\ref{thm:oddW} are proved in Section~\ref{sec:w2d+1fvs}.

\subparagraph*{Treedepth parameterization: class $\XSLP$.}
As we argued, the classes $\XNLP$ and $\XALP$ can be seen as the ``natural home'' for $\csp$ parameterized by pathwidth and treewidth respectively, and for many other problems on ``path-like'' or ``tree-like'' graphs. We introduce a new parameterized complexity class $\XSLP$ which is the ``natural home'' for the parameter treedepth instead, reflecting ``shallow'' graphs (this is what the letter S stands for).

The {\em{treedepth}} of a graph $G$ is the minimum depth of a rooted forest $F$ on the same vertex set as $G$ such that every edge of $G$ connects a vertex with its ancestor in $F$; thus, it is a measure of shallowness of a graph. While treedepth is never smaller than pathwidth, it can be arbitrarily large even on graphs of bounded pathwidth: a path on $n$ vertices has pathwidth $1$ and treedepth $\lceil \log_2 (n+1)\rceil$. Despite being relatively lesser known than treewidth or pathwidth, treedepth appears naturally in many seemingly disconnected areas. For instance, it features a prominent role in the theory of Sparsity (see~\cite[Chapters~6 and~7]{sparsity} for an overview), has interesting combinatorics of its own (see e.g.~\cite{ChenCDFHNPPSWZ21,CzerwinskiNP21,DvorakGT12,KawarabayashiR18}), corresponds to important dividing lines in finite model theory (see e.g.~\cite{ElberfeldGT16,KushR23}), and governs the parameterized complexity of block-structured integer programming (see \cite{monster} for an overview). More importantly for us, a line of work~\cite{FurerY17,HegerfeldK20,NadaraPS22,NederlofPSW23,PilipczukS21,PiWrochna} uncovered that for many classic problems, on graphs of low treedepth one can design fixed-parameter algorithms that are both time- and space-efficient, which is conjectured not to be possible for the pathwidth or treewidth parameterizations~\cite{PiWrochna}. This makes treedepth a prime candidate for a parameter that can be interesting from the point of view of $\csp$.

And so, we define two complexity classes: $\XSLP$ consists of all parameterized problems that can be reduced to $\csp$ parameterized by treedepth in parameterized logspace (that is, in deterministic space $f(k)+\Oh(\log n)$ for a computable $f$), while $\XSLP^+$ has the same definition, except we consider fpt reductions. This distinction is of technical nature: on one hand we use parameterized logspace reductions to match the definitions of $\XALP$ and $\XNLP$ and retain the inclusion $\XSLP \subseteq \XNLP\subseteq \XALP$, and on the other hand we would like to compare $\XSLP$ with the $\W$-hierarchy, which requires closure under fpt reductions. In fact, $\XSLP^+\supseteq \W[t]$ for every integer $t$ (this will follows from Proposition~\ref{prop:S_inclusions}).

We prove two alternative characterizations of $\XSLP$. The first one is through a machine model: we prove that $\XSLP$ can be equivalently defined as problems that can be solved by an alternating Turing machine with the following resource bounds: (1) $f(k)\log n$ bits of nondeterminism, (2) $f(k)+\Oh(\log n)$ bits of conondeterminism, (3) alternation at most $f(k)$, and (4) working space $f(k)+\Oh(\log n)$ plus a read-once stack of size $f(k)\log n$ that can be only pushed upon and read only at the end of the computation. See Theorem~\ref{thm:machine-characterization} in Section~\ref{sec:machine-char} for a formal statement. This reflects the characterization of $\XALP$ through alternating Turing machines with different bounds on conondeterminism and the size of a computation tree, see~\cite[Theorem~1]{BodlaenderGJJPP22}. 

The main step in the proof of our machine characterization of $\XSLP$ is a regularization lemma for the considered machine model, allowing us to assume that the computation tree has always a very concrete shape. Interestingly, this step crucially uses the existence of fpt-sized {\em{universal trees}}, a tool fundamentally underlying the recent advances in the complexity of parity games. While universal trees can be seen only implicitly in the breakthrough work of Calude et al.~\cite{CaludeJKLS22}, their central role in the approach was exposed in subsequent works~\cite{CzerwinskiDFJLP19,JurdzinskiL17}.

The second characterization is through model-checking first-order logic, and is inspired by the definition of the $\A$-hierarchy; see~\cite[Chapter~8]{FlumGrohe}. In essence, we provide a complete problem for $\XSLP$, which amounts to model-checking first-order sentences in which universal quantification must follow a root-to-leaf path in a rooted forest present in the structure. Details and formal statements can be found in Section~\ref{sec:logic-char}.

\subparagraph*{$d$-fold vertex cover and the $\S$-hierarchy.}
Next, we ``project'' the class $\XSLP$ closer to lower levels of the $\W$-hierarchy, thus obtaining a new hierarchy of parameterized classes sandwiched between the $\W$-hierarchy and the $\A$-hierarchy. For this, we introduce the following parameter. 

The \textit{$1$-fold vertex cover number} of a graph $G$ is simply the number of vertices of $G$. Inductively, for $d\geq 2$, the {\em{$d$-fold vertex cover number}} is the smallest integer $k$ with the following property: there is a subset of vertices $U\subseteq V(G)$ with $|U|\leq k$ such that every connected component of $G-U$ has $(d-1)$-fold vertex cover number at most $k$.   
    Alternatively, we can also define the parameter using a ``fattened'' variant of elimination trees (the decomposition notion underlying treedepth). Namely, $G$ has $d$-fold vertex cover number at most $k$ if and only if there is a rooted tree $T$ of depth at most $d$, and a vertex partition $\{V_t\colon t\in V(T)\}$ of $V(G)$ such that $|V_t|\leq k$ for all $t\in V(T)$, and edges in $G$ between vertices of $V_s$ and $V_t$ are only allowed when $s$ and $t$ are equal or are in an ancestor-descendant relationship in $T$.

We now define the parameterized complexity class\footnote{We remark that there is an already existing concept called the $\S$-hierarchy, related to subexponential parameterized algorithms; see~\cite[Definition 16.9]{FlumGrohe}. Since we are not aware of any subsequent work on the structure of this hierarchy, we took the liberty of using the same naming scheme for our classes.} $\S[2d-1]$ as the fpt-closure of $\csp$ parameterized by the $d$-fold vertex cover number, for all integers $d\geq 1$. The following result relates the introduced classes to the $\W$-hierarchy, the $\A$-hierarchy, and the class $\XSLP^+$.
    \begin{restatable}{proposition}{shierarchy}
    \label{prop:S_inclusions}
    For every integer $d\geq 1$,
    we have $\W[2d-1]\subseteq \S[2d-1]\subseteq \A[2d-1]$ and 
    $\S[2d-1]\subseteq \XSLP^+$.
    \end{restatable}
    The proof is straightforward and is given in Section~\ref{sec:st_hierarchy}.

While the definition of $d$-fold vertex cover seems not to have been discussed explicitly in the literature, the idea of alternating deletions of batches of vertices and splitting into connected components is not entirely new, as similar parameters that interpolate between vertex cover and  treedepth have previously been studied. For example, $2$-fold vertex cover is within a multiplicative factor of two of {\em{vertex integrity}}, a parameter that was introduced by Barefoot, Entringer and Swart  \cite{BarefootES87} in 1987 (see \cite{BaggaBGLP92} for a survey). In the context of block-structure integer programs, the {\em{fracture number}}~\cite{Dvorak21EGKO} can be seen as an analogue of $2$-fold vertex cover, while the concept of {\em{topological height}}~\cite{monster} serves a role similar to that of~$d$ in the definition of $d$-fold vertex cover.

\subparagraph*{Comparison to List Coloring.} The classic 
\textsc{List Coloring} problem can be interpreted as the special case of $\csp$ where every constraint just stipulates that the values assigned to adjacent variables are different from each other.
Therefore, a hardness result for \textsc{List Coloring} implies one for $\csp$.
Vice versa, we can attempt to turn an instance of $\csp$ on graph $G$ into an instance of \textsc{List Coloring} by adding, for each edge $uv$ in $G$ and each forbidden pair of values $(a,b)$, a vertex to $G$ adjacent to $u$ and $v$ with color list $\{a,b\}$. This transformation does not significantly affect graph parameters such as treedepth, treewidth or pathwidth, so hardness and completeness results of $\csp$ may also be inherited to \textsc{List Coloring}. However, the transformation may make dramatic changes to other parameters such as vertex cover and vertex modulator to a graph of treedepth at most $d$, where we can only easily deduce $\W[2d-1]$-hardness from our $\W[2d+1]$-hardness results. In fact, we separate the two problems with the following result, proved in Section~\ref{sec:listcol_membership}.
\begin{theorem}
\label{thm:listcol_even}\textsc{List Coloring} is in $\W[2]$ when parameterized by the vertex cover number and in $\W[2d]$ when parameterized by the size of a modulator to a treedepth-$d$ graph.
\end{theorem}
We believe that due to its robustness, $\csp$ better suited to measure the complexity of parameters than \textsc{List Coloring} is. This is also witnessed by the (nearly) tight completeness results presented in Theorems~\ref{thm:vertex_cover},~\ref{thm:fvs}, and~\ref{thm:oddW}. Table~\ref{table:complexityoverview} below presents a comparison of the  parameterized complexity landscapes of $\csp$ and of \textsc{List Coloring} under various structural parameterizations. We postpone the discussion of this table to Section~\ref{section:table_listcol}.

\begin{table}[htb]
    \centering
    \begin{tabular}{c|c|c}
    Parameter     &  Binary CSP & List Coloring  \\\hline
     number of vertices    & W[1]-complete \cite{FellowsHRV09,PapadimitriouY99} & poly-kernel   \\
     vertex cover &  W[3]-complete $\ast$  & W[1]-hard \cite{FialaGK11}, in W[2] $\ast$ \\
     feedback vertex set & W[SAT]-hard, in W[P] $\ast$ & W[3]-hard, in W[P] $\ast$ \\
     modulator to treedepth-$d$ & W[$2d+1$]-complete $\ast$ &    W[$2d-1$]-hard, in  W[$2d$] $\ast$ \\
     modulator to depth $d$-forest & W[$2d+1$]-complete $\ast$ &  W[$2d-1$]-hard, in  W[$2d$] $\ast$ \\ 
     modulator to clique & para-NP-complete & FPT, poly-kernel \cite{BanikJPR20,GutinMOW21} \\
     treedepth    & XSLP-complete $\ast$ & XSLP-complete $\ast$ \\
     tree partition width & XALP-complete \cite{BodlaenderGJJPP22} & W[1]-hard, in XL \cite{BodlaenderGJ22trees} \\
     tree partition width + degree & XALP-complete \cite{BodlaenderGJJPP22} & FPT \\ 
     pathwidth & XNLP-complete \cite{BodlaenderGNS21} & XNLP-complete \cite{BodlaenderGNS21} \\
     bandwidth & XNLP-complete \cite{BodlaenderGNS21}  & FPT \\
     treewidth & XALP-complete \cite{BodlaenderGJJPP22} & XALP-complete \cite{BodlaenderGJJPP22} \\
     treewidth + degree & XALP-complete \cite{BodlaenderGJJPP22} & FPT \\
    \end{tabular}
    \caption{Complexity of $\csp$ and \textsc{List Coloring}. Results marked with $\ast$ are
    shown in this paper. Some results without a reference are easy to obtain. See the 
    discussion in Section~\ref{section:table_listcol}.}
    \label{table:complexityoverview}
\end{table}
\section{Preliminaries}
\label{section:preliminaries}

For integers $a\leq b$, we write $[a,b]$ for $\{a,a+1,\ldots,b\}$.

\subparagraph*{Graphs and their parameters.} 
In this paper, we denote the \emph{depth} of a rooted tree as the maximum number of vertices on a path from root to leaf.
A \emph{rooted forest} is a collection of rooted trees. The \emph{depth} of a 
rooted forest is the maximum depth of the trees in the forest.\footnote{The definitions of depth of a tree used in the literature can differ by one. Here we count the number of vertices, e.g., a tree consisting of a single vertex has depth 1.}

We use standard graph notation.
An \emph{elimination forest} of a graph $G$, is a rooted
forest $F$ with the same vertex set as $G$, such that for each edge $uv$ of $G$, $u$ is an ancestor of $v$
or $v$ is an ancestor of $u$ in $F$. (Note that the forest can contain edges that are not in $G$.)
The \emph{treedepth} of a graph $G$ is
the minimum depth of of rooted forest embedding of $G$.

Let $\mathcal{C}$ be a class of graphs. A \emph{modulator} to $\mathcal{C}$ in a graph $G$ is
a set of vertices $W\subseteq V(G)$, such that the graph $G-W$ belongs to $\mathcal{C}$.
A \emph{vertex cover} of a graph $G$ is a set of vertices $W\subseteq V(G)$, such that every edge of $G$ has at least one
endpoint in $W$. Note that a set of vertices is a vertex cover if and only if it is a modulator to the class of edgeless graphs, or, equivalently, to the class of graphs with treedepth at most $1$.
A \emph{feedback vertex set} in a graph $G$ is a modulator to a forest, or, equivalently,
a set of vertices that intersects each cycle in $G$.

\subparagraph*{Constraint satisfaction problems.} 
We consider the $\csp$ problem defined as follows. An instance of $\csp$ is a triple
$$I=(G,\{D(u)\colon u\in V(G)\},\{C(u,v)\colon uv\in E(G)\}),$$
where
\begin{itemize}
    \item $G$ is an undirected graph, called the {\em{Gaifman graph}} of the instance;
    \item for each $u\in V(G)$, $D(u)$ is a finite set called the {\em{domain}} of $u$; and
    \item for each $uv\in E(G)$, $C(u,v)\subseteq D(u)\times D(v)$ is a binary relation called the {\em{constraint}} at~$uv$. Note that $C(u,v)$ is not necessarily symmetric; throughout this paper, we apply the convention that $C(v,u)=\{(b,a)\,\mid\,(a,b)\in C(u,v)\}$.
\end{itemize}
In the context of a $\csp$ instance, we may sometimes call vertices {\em{variables}}. 
A {\em{satisfying assignment}} for an instance $I$ is a function $\eta$ that maps every variable $u$ to a value $\eta(u)\in D(u)$ such that for every edge $uv$ of $G$, we have $(\eta(u),\eta(v))\in C(u,v)$. The $\csp$ problem asks, for a given instance $I$, whether $I$ is {\em{satisfiable}}, that is, there is a satisfying assignment for $I$.

The \textsc{List Coloring} problem is a special case of $\csp$ defined as follows. An instance consists of a graph $G$ and, for every vertex $u$ of $G$, a set ({\em{list}}) of colors $L(u)$. The question is whether there is a mapping $f$ of vertices
to colors such that for every vertex $u$ we have $f(u)\in L(u)$, and for each edge $uv$ of $G$, we have $f(u)\neq f(v)$. Note that this is equivalent to a $\csp$ instance where lists $L(u)$ are the domains, and all constraints are non-equalities: $C(u,v)=\{(a,b)\in L(u)\times L(v)~|~a\neq b\}$ for every edge $uv$.

\subparagraph*{Complexity theory.} 
We assume the reader to be familiar with standard notions of the parameterized complexity theory, such as the W-hierarchy or parameterized reductions. For more background, see ~\cite{CyganFKLMPPS15,DowneyF99,DowneyF13,FlumGrohe}. Let us recall concepts directly used in this paper.

We say that a parameterized problem $Q$ is in {\em{parameterized logspace}} if $Q$ can be decided in (deterministic) space $f(k)+\Oh(\log n)$, for some computable function $f$. Note that every problem in parameterized logspace is fixed-parameter tractable, because a Turing machine working in space $f(k)+\Oh(\log n)$ has $2^{\Oh(f(k))}\cdot n^{\Oh(1)}$ configurations, and hence its acceptance can be decided in fixed-parameter time. 

An {\em{fpt-reduction}} is a parameterized reduction that works in fixed-parameter time. A {\em{pl-reduction}} is a parameterized reduction that works in parameterized logspace, that is, can be computed in (deterministic) space $f(k)+\Oh(\log n)$, for some computable function $f$.

A Boolean formula is said to be {\em{$t$-normalized}} when it is the conjunction of disjunctions of conjunctions of \ldots of literals,
with $t$ levels of conjunctions or disjunctions. We only consider the case where $t\geq 2$, and assume that we start by
conjunctions. Note that 2-normalized Boolean formulas are in Conjunctive Normal Form.

In the \textsc{Weighted $t$-Normalized Satisfiability} problem, we are given a $t$-normalized 
Boolean formula $F$ on $n$ variables, and an
integer $k$,
and ask if we can satisfy $F$ by setting exactly $k$ of the variables to true, and all other variables to false.
This problem is complete for $\W[t]$, see e.g. \cite{DowneyF99,DowneyF13}.
A $t$-normalized expression is said to be \emph{anti-monotone} if each literal is the negation of a variable.
We use the following result to simplify our proofs.
\begin{theorem}[Downey and Fellows, see \cite{DowneyF99,DowneyF13}]
For every odd $t\geq 3$, \textsc{Weighted Anti-Monotone $t$-Normalized Satisfiability} is complete for $W[t]$.
\label{theorem:df2}
\end{theorem}
We use the following result as starting point for membership proofs.

\begin{theorem}[Downey and Fellows, see \cite{DowneyF99,DowneyF13}]
For every $t\geq 2$, \textsc{Weighted $t$-Normalized Satisfiability} is complete for $\W[t]$.
\label{theorem:df}
\end{theorem}

\section{W[3]-completeness for BinCSP parameterized by vertex cover}
\label{section:hardness}
In this section, we prove Theorem \ref{thm:vertex_cover}. We prove hardness and  membership in two separate lemmas below (Lemma~\ref{lemma:w3hard} and Lemma~\ref{lem:vc_membership}).  

\begin{lemma}
        \label{lemma:w3hard}
$\csp$ with vertex cover as parameter is $\W[3]$-hard.
\end{lemma}

\begin{proof}
    Take an instance of \textsc{Weighted 3-Normalized Anti-Monotone Satisfiability}, i.e.,
    we have a Boolean formula $F$ that is a conjunction of disjunctions of conjunctions of negative literals, and ask if we
    can satisfy it by setting exactly $k$ variables to true. Suppose $x_1, \ldots, x_n$
    are the variables used by $F$. Suppose $F$ is the conjunction of $r$ disjunctions of conjunctions of negative literals.

    We build a graph $G$ as follows. The vertex set $V(G)$ consists of a set $W= \{w_1, \ldots, w_k\}$ of size $k$, and a set $S = \{v_1, v_2, \ldots, v_r\}$
    of size $r$. The set $W$ will be the vertex cover of $G$, and $S$ will form an independent set. 
    We add edges from each vertex in $W$ to each other vertex in the graph. 

    The domain of a vertex $w\in W$ is $D(w) = \{x_1, \ldots, x_n\}$. For distinct $w,w'\in W$, $w'\neq w'$, we set $C(w,w') = \{(x_i,x_j)~|~ i\neq j\}$. This enforces that all vertices in $W$ are assigned a different
    value --- this corresponds to setting exactly $k$ variables to true.
    
    Now consider a vertex $v_i\in S$ for $i\in [1,r]$.
    We say that $v_i$ represents the $i$th disjunction of conjunctions of literals in $F$, i.e.,
    each of the disjunctions in the formula is represented by one vertex in the independent set.
    Suppose that this disjunction has $t_i$ terms (each term is a conjunction of negative
    literals). 
    We set $D(v_i) = [1,t_i]$, that is, each value for $v_i$ is an integer in $[1,t_i]$.

    The intuition is as follows. We set a variable $x_i$ to true, if and only if exactly
    one vertex in $W$ is assigned $x_i$. As all vertices in $W$ will get a different value,
    we set in this way exactly $k$ variables to true. 
    The formula $F$ is the conjunction of $r$ disjunctions; each of these disjunctions is represented by
    one of the vertices $v_i\in S$. For each $v_i$, the disjunction represented by $v_i$
    must be satisfied, so one of its terms must be satisfied. The value of $v_i$ tells 
    a satisfied term, i.e., if the value of $v_i$ is $j\in [1,t_i]$, then the $j$th term
    is satisfied. This is checked by looking at the edges from $v_i$ to the vertices in $W$.

    We now give the constraints that ensure the term is satisfied. Consider a vertex $v_i\in S$ and $w\in W$. Recall that the value of $v_i$ is an integer in $[1,t_i]$ which represents one term in the $i$th disjunction of $F$, and that term is a conjunction of a number of negative
    literals.
    For $j\in [1,t_i]$ and $j'\in [1,n]$, we have $(j,x_{j'})\in C(v_i,w)$ if and only
    if for each literal $\neg x_{j''}$ that appears in the $j$th term of the $i$th disjunction
    of $F$, $j'' \neq j'$.
   
We call the constructed graph $G$ and write $I$ for the corresponding instance of $\csp$.
\begin{claim}
    \label{cl:membershipw3}
    $F$ can be satisfied by setting exactly $k$ variables to true, if and only if 
    $I$ has a satisfying assignment.
\end{claim}

\begin{claimproof}[Proof of Claim~\ref{cl:membershipw3}]
Suppose $F$ can be satisfied by making $x_{i_1}, \ldots, x_{i_k}$ true, and all other literals false.
Then assign the vertices in $W$ the values  $x_{i_1}, \ldots, x_{i_k}$ successively.
The constraints between vertices in $W$ are thus satisfied.

Now consider a vertex $v_i\in S$. Consider the $i$th term $F_i$ of the (upper level) conjunction of $F$. This term must be
satisfied by the truth assignment. Suppose the term is $F_i = F_{i,1} \vee \cdots \vee F_{i,t_i}$. At least one of the $F_{i_j}$'s must
be satisfied by the truth assignment, say $F_{i,j'}$. Then assign $v_i$ the value $j'$.

We can verify that the constraints for edges between $v_i$ and each $w_j$ are fulfilled. By assumption,
$F_{i,j'}$ holds. It thus cannot contain a negative literal $\neg x_\alpha$, where $x_\alpha$ is set to true.
So $w_j$ cannot be assigned $x_\alpha$ when $\neg x_\alpha$ is a literal in $F_{i,j'}$. 
Thus we found a satisfying assignment for $I$.

\medskip

Now, suppose that $I$ has a satisfying assignment. From the constraints between vertices in $W$,
we see that all vertices in $W$ have a different value. Set a variable $x_i$ to true, if and only if a vertex in $W$
has value $x_i$, and otherwise, set it to false. We have thus set exactly $k$ variables to true.

Consider the $i$th term of the upper level conjunction of $F$. Suppose this term is $F_{i,1}\vee \ldots \vee F_{i,t_i}$. 
Suppose $v_i$ is assigned value $j$.
For each negative literal $\neg x_\alpha$ in the conjunction~$F_{i,j}$, by the constraints, we cannot
have a vertex in $W$ that is assigned $x_\alpha$, and thus $x_\alpha$ is set to false.
Thus,
the term $F_{i,j}$ is satisfied by the truth assignment, and thus $F_i$ is satisfied. As this holds for all conjuncts of $F$, $F$ is satisfied by the specified assignment.
\end{claimproof}
From Claim~\ref{cl:membershipw3}, we see that we have a parameterized reduction from \textsc{Weighted Anti-Monotone 3-Normalized Satisfiability}
to $\csp$ with vertex cover as parameter. The result
now follows from the $\W[3]$-hardness of \textsc{Weighted Anti-Monotone 3-Normalized Satisfiability} (Theorem \ref{theorem:df2}).
\end{proof}

\begin{restatable}{lemma}{vcmember}
\label{lem:vc_membership}
    $\csp$ parameterized by the vertex cover number is in $\W[3]$.
\end{restatable}
\begin{proof}
Suppose we are given an instance $I=(G,\{D(u)\colon u\in V(G)\}, \{C(u,v)\colon uv\in E(G)\})$ of $\csp$. Let $W\subseteq V(G)$ be a vertex cover of $G$ of size $k$; such a vertex cover can be computed in fpt time.
Write $W = \{w_1, w_2, \ldots, w_k\}$.

We build a formula $F$ as follows. For each vertex $w\in W$ and each
value $c\in D(w)$, we make a Boolean variable $x_{w,c}$. 
The intuition is that $x_{w,c}$ is true, iff we assign $c$ to $w$. Write $X=\{x_{w,c}\colon w\in W, c\in D(w)\}$ for the variable set.

We first give a formula $F^1(X)$, that ensures that each vertex has at least one value.
\[
F^1(X) = \bigwedge_{w\in W} \bigvee_{c\in D(w)} x_{w,c}
\]
Note that when we set exactly $k$ variables to true and $F^1(X)$ holds, then for each $w\in W$, there is exactly one
$c\in D(w)$ with $x_{w,c}$ true. We call this assignment of the vertices in $W$ the \emph{assignment given by $X$}.

In the second step, we verify that the assignment given by $X$ does not create a conflict between vertices in $W$.
\[
F^2(X) = \bigwedge_{w_1,w_2\in W, w_1w_2\in E(G)}\ \bigvee_{(c,c')\in C(w_1,w_2)} x_{w_1,c} \wedge x_{w_2,c'}
\]
Note that, assuming $F^1(X)$ also holds, that $F^2(X)$ holds, if and only if the assignment given by $X$ does not create
a conflict between vertices in $W$.

Our third formula has as argument a vertex $v\in V(G)\setminus W$, and a value $c\in D(v)$, and checks if we can assign
$c$ to $v$ without creating a conflict with the assignment given by $X$. To keep the formula a single conjunction,
we check for each vertex in $w$ and each value $c'$ for $w$ that would conflict with $c$ that $w$ does not have value $c'$.

\[
F^3(X,v,c) = \bigwedge_{w\in W: vw\in E(G)}\ \bigwedge_{c'\in D(w): (c,c')\not\in C(v,w)} \neg x_{w,c'}
\]

Formula $F^4$ checks that all vertices in $V(G)\setminus W$ can be assigned a value without creating a conflict with the assignment given by $X$.
\[
F^4(X) = \bigwedge_{v\in V(G)\setminus W}\ \bigvee_{c\in D(v)} F^3(X,v,c)
\]
Finally, we define $F(X) = F^1(X) \wedge F^2(X) \wedge F^4(X)$.

From the discussion above, we see that $F(X)$ holds, if and only if $X$ gives an assignment of the vertices in $W$ that
does not create a conflict between vertices in $W$, and each vertex in $V(G)\setminus W$ can choose a value without creating conflicts
between vertices in $W$ and vertices in $V(G)\setminus W$. As $W$ is a vertex cover, $F(X)$ holds, if and only if the assignment of $W$
given by $X$ can be extended to the entire graph, i.e., if and only if the $\csp$ instance
has a solution.

As each term $F^3(X,v,c)$ is a conjunction of literals, $F^4$ and $F$ are 3-normalized. Observe that they have 
polynomial size. We can now conclude the result.
\end{proof}

\section{\texorpdfstring{$\W[2d+1]$}{W[2d+1]}-completeness and feedback vertex set}
\label{sec:w2d+1fvs}
We prove
Theorem~\ref{thm:oddW} by proving the $\W[2d+1]$-hardness and $\W[2d+1]$-membership in Corollary~\ref{cor:w2d+1hard-treedepth} and Lemma~\ref{lemma:in} respectively and prove the hardness and membership results about $\csp$ parameterized by feedback vertex set (Theorem \ref{thm:fvs}) in Lemma \ref{lemma:fvs-hardness} and Lemma~\ref{lem:fvsmem}.

\subsection{Hardness results}
We first prove our hardness results.
\begin{lemma}
    \label{lemma:w2d+1hard}
Let $d\geq 2$ be an integer.
    $\csp$ with the size of a modulator to forests of depth $d$ as parameter is $\W[2d+1]$-hard.
\end{lemma}
\begin{proof}
    Take an instance of \textsc{Weighted $(2d+1)$-normalized Anti-Monotone Satisfiability}, i.e.,
    we have formula $F$ that is a conjunction of disjunctions of conjunctions of \ldots of conjunctions of negative literals, with $d+1$ levels of conjunction, and $d$ levels of disjunction. 
    Suppose the variables used in $X$ are $\{x_1, \ldots, x_n\}$.

Represent $F$ as a rooted tree $T=(V,E)$, with the root representing $F$, and
each non-leaf representing a disjunction or conjunction, with the terms as children, and the
literals as leaves.
The depth of this tree is $2d+2$. We also assume that $T$ is ordered, i.e., for each vertex with children, 
we have an ordering on the children (which thus allows us to talk about the $i$th child of a vertex).
For each vertex $v\in V$, write $F_v$ as the subformula of $F$ that corresponds to $v$.

Number the levels of the tree, with the root level 1, the children of the root level 2, etc.
Note that the nodes in $T$ on odd levels represent a conjunction, and the nodes in $T$ on
even levels represent a disjunction, except for the lowest level (which is $2d+2$, thus even),
where the nodes represent literals.
(By definition of normalized, all leafs are at level $2d+2$.)
Now, build a forest $T'$, obtained by removing the root of $T$, and contracting all other
nodes in $T$ on odd levels to their parent. I.e., only the nodes on even levels of $T$ remain,
and each is a child of their grandparent in $T$. Finally, remove all leaves (all nodes representing
a literal).
Note that $T'$ is a forest with each tree of depth $d$. 

The graph $G$ is formed by taking the obtained forest $T'$, and adding a set $W$ with $k$ vertices,
with an edge from $W$ to each leaf in $T'$. Write $W=\{w_1, \ldots, w_k\}$.
Also, we turn $W$ into a clique.

The role of the values of $W$ are as in the proof of Lemma~\ref{lemma:w3hard}. For each $w\in W$, set
$D(w) = \{x_1, \ldots, x_n\}$, and for $w$, $w'\in W$, $w\neq w$, set
$C(w,w') = \{ (x_i,x_{i'})~|~i\neq i'\}$.

We distinguish three types of vertices in $T'$. First, $R$ is the set of vertices in $T'$ that
are the root of a tree in the forest. Second, $L$ is the set of vertices in $T'$ that are a leaf.
Finally, the set of other vertices in $T'$ (i.e., those that are not a root or leaf) is called $O$.
Note that $F= \bigwedge_{v\in R} F_v$.

Each vertex in $T'$ represents a disjunction. For a vertex $v$, let $t_v$ be the number of
terms of the disjunction, represented by $v$. (This equals the number of children of $v$ in
$T$, but in general, $v$ will have much more children in $T'$.)
The domains of vertices in $T'$ are:
\begin{itemize}
    \item If $v\in R$, then $D(v) = [1,t_v]=\{1,\dots,t_v\}$, used to point to a term of $F_v$. The formula $F_v$ is a disjunction, and the value of $v$ tells a
    term of this disjunction that is satisfied.
    \item If $v\in O\cup L$, then $D(v) = [1,t_v] \cup \{\square\}$. The value either again is
    an integer in $[1,t_v]$ 
    that points to a satisfied term of the disjunction $F_v$,
    or is the value $\square$, which stands for {\em inactive}. The inactive value means that  
    $F_v$ is not necessarily satisfied (it can be satisfied or not satisfied.) A vertex with a value in $[1,t_v]$ is said to be \emph{active}. 
\end{itemize}
Roots of the subtrees cannot have the value $\square$, so are always active.

For edges between vertices in $L$ and $W$, the
constraints are similar as in the proof of Lemma~\ref{lemma:w3hard}, except that we also allow all pairs where $v$ is inactive.
Let $v\in L$ and $w\in W$. Note that $F_v$ is a disjunction
of conjunctions of negative literals, say $F_v = \bigvee_{i\in [1,t_v]} F_{v,i}$,
with $F_{v,i}$ a conjunction of negative literals. For $i\in [1,t_v]$ and $j\in [1,n]$, we have
$(i,x_j)\in C(v,w)$, if and only if $\neg x_j$ is not part of the conjunction $F_{v,i}$.
In addition, we have $(\square,x_j)$ in $C(v,w)$ for all $j\in [1,n]$.

Now, we consider an edge $vv'\in E$, with $v,v'\in R\cup O\cup L$. Suppose $v$ is the parent of $v'$. The set $C(v,v')$ consists of all pairs fulfilling one of the following conditions.
\begin{itemize} 
\item If $v\not\in R$, then $(\square,\square)$ is in $C(v,v')$.
\item For $i\in [1,t_v]$, $v'$ not a child of the $i$th child of $v$ in $T$, 
then $(i,\square)$ is in $C(v,v')$.
\item For $i\in [1,t_v]$, $i'\in [1,t_{v'}]$, $v'$ a child of the $i$th child of $v$ in $T$,
then $(i,i')$ is in $C(v,v')$.
\end{itemize}
Let $I$ denote the corresponding instance of $\csp$.
\begin{claim}
\label{cl:woddhardness}
    $I$ has a satisfying assignment if and only if $F$ can be satisfied by setting exactly $k$ variables to true.
\end{claim}
\begin{claimproof}
Suppose $F$ can be satisfied by setting exactly $k$ variables to true. Now,
    for each true variable, assign that variable to a vertex in $W$.

    We say that all root nodes of trees in $T'$ are active. Each active vertex is a
    disjunction that is satisfied by the variable setting. Top-down, we assign values, as follows.
    If a vertex $v$ is active, then at least one of the terms of the disjunction $F_v$ must
    be satisfied. Choose a $s_v\in [1,t_v]$ such that the ($s_v$)th term of $F_v$, with $v$ active,
    is satisfied. Assign $s_v$ to $v$. All vertices $v'$ that are in $T$
    a child of the ($s_v$)th child of $v$ are said to be active. These vertices precisely are
    the terms of the conjunction that is represented by the ($s_v$)th child of $v$ --- as we
    assume that that child represents a satisfied conjunction, all its terms are also satisfied.
    Thus, by induction (top-down in the tree), active vertices represent satisfied subformulas.
    An inactive vertex is assigned $\square$.

    One easily checks that all edges between vertices in $T'$ fulfil the constraints.

    An edge between a vertex $v\in L$ and a vertex $w\in W$ always satisfies its constraints when
    $v$ has the value $\square$. If $v$ has a value $s_v \in [1,t_v]$, then $v$ is active,
    hence $F_v$ satisfied, and the $(s_v)$ term of $F_v$ is satisfied. Thus, all literals in that
    term (which is a conjunction) are satisfied, and so the term cannot contain a literal
    $\neg x_j$ with $x_j$ set to true. So, $w$ cannot have the value $x_j$.

    \medskip

    Now, suppose $I$ has an assignment satisfying all constraints. The constraints on $W$ give
    that each vertex in $W$ has a different value from $\{x_1, \ldots, x_n\}$. Set $x_i$ to
    true, if and only if there is a vertex in $W$ with value $x_i$. So, we have set $k$ variables
    to true.

    We claim that for each active vertex $v$ in $T$ (i.e., a vertex with value different from
    $\square$), $F_v$ is satisfied by the defined truth assignment.
    We proof this by induction, bottom-up in the tree. 
    If $v\in L$ is a leaf with value $s_v\in [1,t_v]$, then the $(s_v)$th term of $F_v$ is satisfied: this term is a conjunction of negative literals. Consider
    a literal $\neg x_i$ that appears in the term. If $x_i$ would be true, then there is
    a vertex $w\in W$ with value $x_i$, but then the edge $(v,w)$ would not satisfy the 
    constraints. So, the $(s_v)$th term of $F_v$ is a conjunction of satisfied literals, and thus the disjunction $F_v$ is satisfied.
    If $v\in X\cup R$ with value $s_v\in [1,t_v]$, then the ($s_v$)th term of $F_v$ is a conjunction of terms. The  constraints enforce that for all these
    terms their corresponding vertex is active (they are children of the ($s_v$)th child of
    $v$ in $T$), and thus, the ($s_v$)th term of $F_v$ is satisfied, and thus $F_v$ is satisfied.

    As roots of trees in $T'$ are always active (cannot have value $\square$), for each
    root $v$ of a tree in $T'$, we have $F_v$ is satisfied, so $F=\bigwedge_{v\in R} F_v$ is satisfied.
\end{claimproof}
As $G$ and its assignment can be constructed in polynomial time, and we keep the same 
parameter $k$, the result follows from the fact that \textsc{Weighted $t$-Normalized
Anti-Monotone Satisfiability} is complete for $\W[t]$ (Theorem \ref{theorem:df2}).
\end{proof}
Since a forest of depth $d$ has treedepth at most $d$, we directly obtain the following corollary.
\begin{corollary}
    \label{cor:w2d+1hard-treedepth}
        $\csp$ with the size of a modulator to a graph of treedepth at most $d$ as the parameter  is $\W[2d+1]$-hard.
\end{corollary}
We also easily deduce the result for feedback vertex set from the proof of  Lemma~\ref{lemma:w2d+1hard}.
\begin{lemma}
    $\csp$ parameterized by the feedback vertex number
    is $\WSAT$-hard.
    \label{lemma:fvs-hardness}
\end{lemma}

\begin{proof}
\textsc{Weighted Anti-Monotone Satisfiability} is $\WSAT$-complete \cite{AbrahamsonDF95}.
    We repeat the proof of Lemma~\ref{lemma:w2d+1hard}, but instead
    use \textsc{Weighted Anti-Monotone Satisfiability} as problem to reduce from, and
    have no depth bound on the resulting tree or forest. 
\end{proof}

\subsection{Membership results}
Next, we prove our membership results.
\label{sec:in}
\begin{lemma}
    $\csp$ with the size of a modulator to a treedepth-$d$ graph is in $\W[2d+1]$.
    \label{lemma:in}
\end{lemma}
\begin{proof}
Assume an instance $I=(G,\{D(u)\colon u\in V(G)\},\{C(u,v)\colon uv\in E(G)\})$ of $\csp$, with a modulator set of vertices
$W$ of size $k$, and an elimination forest $T$ of $G-W$ of depth $d$. (Note that such a modulator $W$ can be computed in fpt time using an application of Courcelle's theorem, while an elimination forest $T$ as above can be computed in fpt time using the algorithm of Reidl et al.~\cite{ReidlRVS14}.)
Write $W= \{w_1, \ldots, w_k\}$.
We can assume that all edges of $T$ are also present in $E(G)$, otherwise, add the edge
to $E(G)$; for such a new edge $yz$, set $C(y,z) = D(y)\times D(z)$. (By allowing all pairs on new edges, the collection of valid assignments does not change.)

As in the proof of Lemma \ref{lem:vc_membership}, we take a Boolean variable $x_{w,c}$ for each $w\in W$, $c\in D(w)$. The first two steps are identical to that proof. We define
\[ F^1(X) = \bigwedge_{w\in W} \bigvee_{c\in D(w)} x_{w,c} \] 
and
\[ F^2(X) = \bigwedge_{w_1,w_2\in W, w_1w_2\in E(G)}\ \bigvee_{(c,c')\in C(w_1,w_2)} x_{w_1,c} \wedge x_{w_2,c'}. \]

If $F^1(X)$ holds, then again for each $w\in W$, there is exactly one value $c\in D(w)$ with $x_{w,c}$ true; the thus obtained assignment of $W$ is again called the assignment given by $X$, 
and $F^2(X)$ holds, if and only if this assignment does not give a conflict between vertices in $W$.

For each $v\in V(G)\setminus W$, let $A(v)$ be the set of ancestors of $v$, including $v$ itself; and let $B(v)$ be the set of children of $v$.
Note that for each $v\in V(G)\setminus W$, $|A(v)|\leq d$, as $T$ is a forest of depth at most $d$.

Let $f\colon A(v) \rightarrow \mathcal{C}$ be a function that assigns a value to each vertex in $A(v)$.
We say that $f$ is \emph{conflict-free}, iff for all $x\in A(v)$, $f(x)\in D(x)$, and
for all $x,y\in A(v)$ with $xy\in E(G)$, $(f(x),f(y)\in C(x,y)$, i.e., the assignment satisfies locally
the given constraints of the CSP instance.

For a vertex $v\in V(G)\setminus W$, we let $\mathcal{F}_v$ be the set of all conflict-free functions $f\colon A(v) \rightarrow \mathcal{C}$.
Note that these sets do not depend on $X$ and can be computed in polynomial time for each vertex.

The next term, $F^5(X,v,f)$ is defined for a $v\in V(G)\setminus W$ and $f\in \mathcal{F}_v$. It holds, if and only if the
assignment $f$ does not create a conflict with the assignment of $W$ given by $X$:
\[
F^5(X,v,f) = \bigwedge_{x\in A(v), w\in W}\ \bigwedge_{c\in D(w), (f(x),c)\not\in C(x,w)} \neg x_{w,c}.
\]

Suppose $x$ is a child of $v$ in $T$, and $f\in \mathcal{F}_v$, $c\in C(x)$. We define
$f\oplus (x\mapsto c)$ to be the function, that extends the domain of $f$ with the element $x$ and maps $x$ to $c$.
Consider the following, recursive definition. It is defined for a $v\in V(G)\setminus W$ and $f\in \mathcal{F}_v$. We set
\[
F^6(X,v,f) = F^5(X,v,f) \wedge \bigwedge_{y\in B(v)}\ \bigvee_{c\in C(y), (f\oplus (x\mapsto c))\in \mathcal{F}_y } F^6(X,y,f\oplus (x\mapsto c)).
\]
The formula
$F^6(X,v,f)$ holds, if and only if there is an assignment without conflicts for $v$, the ancestors of $v$,  the descendants of $v$, and the vertices in $W$, when $X$ gives the values for the vertices in $W$, and $f$ gives the values for $v$ and its ancestors.
Or, in other words, if we can extend the assignment defined by $X$ and $f$ to the descendants of $v$.
That this is indeed the property expressed by $F^6$ can be shown with induction to the maximum distance to a leaf of vertices.
For leaves, $F^6(X,v,f) =  F^5(X,v,f)$ and the property holds as leaves have no descendants.
For non-leaves, the property holds if we can find a value for each child in $B(v)$ and for each vertex in the subtree below the child.

The next formula checks that we can assign a value to all vertices in $V(G)\setminus W$, given the assignment of $W$ given by $X$.
For this, we check that we can choose a value for each root, such that this value can be extended to
an assignment of the subtree below the root. Let $R$ be the set of roots of $T$. Note that a
set $A(r)$ for $r\in R$ has size 1. We set
\[
F^7(X) = \bigwedge_{r\in R} \bigvee_{f\in \mathcal{F}_r} F^6(X,r,f).
\]
Finally, we can define $F$:
\[ F(X) = F^1(X) \wedge F^2(X) \wedge F^7(X).\]
The formula $F(X)$ holds, if and only if the assignment of $W$ given by $X$ can be extended to a assignment of the entire graph while
satisfying all constraints. What remains to be shown is that $F$ is $(2d+1)$-normalized, and that it can be computed
in polynomial time.

First, notice that $F^1$ and $F^2$ are 3-normalized (where we sometimes take a single term as a conjunction or disjunction of 1 term).
If $v$ is a leaf, then each term of the form $F^5(X,v,f)=F^6(X,v,f)$ is a conjunction of literals, i.e., 1-normalized.
With induction, we have that if the maximum distance of a vertex $v$ to a leaf is $d'$, then each term
$F^6(X,v,f)$ is ($2d'+1$)-normalized. As for each root, the distance to a leaf is at most $d-1$, each term
$F^6(X,v,f)$ is ($2d-1$)-normalized, and thus $F^7$ and $F$ are ($2d+1$)-normalized.

Second, we note that we can compute 
$F^1$, $F^2$, and all terms of the form $F^5$ in polynomial time, and
thus also all terms of the form $F^6$ for leaves of $T$. Suppose we have $n$ vertices, and $N= |\mathcal{C}|$.
For each $v\in V(G)\setminus W$, $\mathcal{F}_v$ has at most $N^d$ elements, and as $d$ is constant here, we can
build all sets $\mathcal{F}_v$ in polynomial time.
In order
to build one term $F^6(X,v,f)$, we need to build less than $n\cdot N$ terms of the form $F^6(X,y,f')$ for children $y$ of $v$.
For each of these, the maximum distance of $y$ to a leaf is one smaller than the maximum distance of $v$ to a leaf.
Thus, if the maximum distance of $v$ to a leaf is $d'$, then the time to compute one term $F^6(X,v,f)$ is $n^{2d'+\Oh(1)} \cdot N^{2d'+\Oh(1)}$. The maximum distance of a root to a leaf is $d-1$, and $F^7$ has to compute at most $|R|\cdot N$ terms,
so the total time to build $F$ is bounded by $\Oh(n^{2d+\Oh(1)}N^{2d+\Oh(1)})$.

Thus, we have a parameterized reduction from $\csp$ with modulator to treedepth-$d$ to
\textsc{Weighted $(2d+1)$-Normalized Satisfiability}. Membership of the latter in $\W[2d+1]$ 
(cf. Theorem~\ref{theorem:df}) gives the result.
\end{proof}
As forests of depth $d$ have treedepth $d$, we also have the following result. (A direct
proof of this fact can be simpler than the proof above, as we can avoid the use of sets $A(v)$.)
\begin{corollary}
    \label{corollary:in2}
    $\csp$ with the size of a modulator to a forest of depth $d$ as 
    parameter is in $\W[2d+1]$.
\end{corollary}
For the parameter feedback vertex set, we cannot obtain $\WSAT$-membership in a similar fashion. As in the proof above, we can create an equivalent formula of size $n^{\Oh(d)}$ for $n$ the input size of the instance and $d$ the depth of the corresponding Gaifman graph, but now that $d$ is not bounded, this may become too large. Instead, we are able to give a circuit formulation since this allows us to `reuse' terms in the formula in a similar fashion to dynamic programming.
\begin{lemma}
\label{lem:fvsmem}
    $\csp$ with the feedback vertex number as parameter is
    in~$\WP$.
\end{lemma}

\begin{proof}
    We can show that a problem is in $\WP$ by giving a circuit, that has as input $n$ Boolean
variables, that has polynomial size, and has an accepting input with exactly $k$ variables
set to true, if and only if the instance of $\csp$ has a solution.

Suppose we are given an instance $(G,\{D(u)\colon u\in V(G)\},\{C(u,v)\colon uv\in E(G)\})$ of $\csp$.
Suppose $W$ is a minimum feedback vertex of $G$; note that $W$ can be computed in fpt time, see~\cite{CyganFKLMPPS15}. Write $G- W$ as forest $T$, and choose in
each tree in $T$ a root.

For each $w\in W$, and $c\in D(w)$, we have a variable $x_{w,c}$, that corresponds to
assigning $c$ to $w$. The circuit consists of two parts. One part has an output gate that is
true, iff each vertex in $W$ gets exactly one value, and the assignment is conflict-free (similar to $F^1$ and $F^2$ of the proof of Lemma \ref{lemma:in}); the other part has an output gate that is
true, iff the assignment of $W$ can be extended to the vertices in $T$ without creating a conflict.
A final and-gate has as input the outputs of these parts, and gives the output of the circuit.

We now describe the second part. For $v\in V(G)\setminus W$
a vertex in $T$, and $c\in D(v)$, define $y_{v,c}$ to be true, if and only if there
is an assignment of the subgraph, consisting of $v$, $W$, and the descendants of $v$ fulfilling
all constraints, that assigns $c$ to $v$, and assigns the vertices
 in $W$ as dictated by the setting
of the  input variables. Let $B(v)$ be the children of $v$ in $T$.
One easily checks that the following recursive definition of $y_{v,c}$ is correct, 
and acyclic:
\[
y_{v,c} =\left(\bigwedge_{w\in W} \bigvee_{(c,c')\in C(v,w)} x_{w,c'}\right) \wedge \left(
\bigwedge_{v'\in B(v)} \bigvee_{(c,c')\in C(v,v')} y_{v',c'}\right).
\]

Thus, we can build a circuit that computes all values $y_{v,c}$.
The second part then outputs, with $R$ the set of roots in $T$
\[
\bigwedge_{v\in R} \bigvee_{c\in D(v)} y_{v,c}. 
\]
For $v\in R$, $c_1,c_2\in D(v)$, $v'\in B(v)$ a child of $v$ and $c'\in D(v')$ with $(c_1,c'),(c_2,c')\in C(v,v')$, the circuit will compute $y_{v',c'}$ only once and will throughput it to compute both $y_{v,c_1}$ and $y_{v,c_2}$. This is what allows us to create a polynomial-size circuit even though the formula above is too large.
\end{proof}
\section{W[2]-membership for List Coloring}
\label{sec:listcol_membership}
In this section we prove Theorem \ref{thm:listcol_even}, showing that 
\textsc{List Coloring} has a different place in the W-hierarchy than $\csp$ has, for the parameters vertex cover, modulator to depth-$d$ forest of depth $d$ and modulator to treedepth-$d$.
\begin{lemma}
\label{lem:listcol_vc_membership}
    \textsc{List Coloring} parameterized by the vertex cover number is in $\W[2]$.
\end{lemma}

\begin{proof}
Suppose that we are given an instance of \textsc{List Coloring}, with a
graph $G=(V,E)$ and color lists $L(v)\subseteq \mathcal{C}$ for all $v\in V$. Suppose $W\subseteq V$ is a vertex cover of $G$ of size $k$; such a vertex cover can be computed in fpt time.
Write $W = \{w_1, w_2, \ldots, w_k\}$. We fix a total order $<$ on the elements of $\mathcal{C}$.

We aim to find a formula $F$ on some variable set $X$, that is a conjunction of disjunctions of literals, such that the formula can be satisfied by setting exactly $k'$ variables to true for some $k'\leq 4\cdot (2^k+1)$ that will be specified later. 

We will first create a formula for an instance in which each vertex in $V\setminus W$ is adjacent to all vertices in $W$, and later explain how to use this in the general case. 

As in the proof of Lemma \ref{lem:vc_membership}, we will use the variables $x_{w,c}$ for $w\in W$ and $c\in L(w)$ to model a coloring $\alpha$ of $W$ and will then verify for each vertex $v\in V\setminus W$ that we can assign $v$ a color from $L(v)\setminus \{\alpha(w_1),\dots,\alpha(w_k)\}$. This can be done with a formula of the form $\wedge_{v\in V\setminus W}\vee_{c\in L(v)}\wedge_{w\in W} \neg x_{w,c}$, but we are not allowed to use that last conjunction. 

Suppose for simplicity that $\mathcal{C}=\{1,\dots,m\}$ and $\alpha(w_1)\leq \dots\leq \alpha(w_k)$. Then for $c\in L(v)$, we find $c \not \in \{\alpha(w_1),\dots,\alpha(w_k)\}$ if and only if $c$ is an element in the union of `intervals'
\[
\{1,\dots, \alpha(w_1)-1\}\cup \left(\bigcup_{i=1}^{k-1}\{\alpha(w_i)+1,\dots,\alpha(w_{i+1})-1\}\right) \cup \{\alpha(w_k)+1,\dots, m\}. 
\]
We can use further auxiliary vertices and variables in order to `sort' the colors, and store the `endpoints' of these intervals, such that we can easily check whether the color is in one of the desired intervals. 

We create $3k$ further vertices $w_1',\dots, w_{3k}'$ and set $W'=\{w_1',\dots,w_{3k}'\}$. These are not part of the original instance of \textsc{List Coloring}, but will be used to define the formula $F$. 
We will define color sets $L(w')$ for each $w'\in W'$ and color constraints $C(w,w')\subseteq L(w)\times L(w')$ for distinct $w,w'\in W\cup W'$ later. When we choose adjacent vertices from $W$, the constraint will simply be that the vertices receive different colors, and the other constraints will be used to create the `intervals' mentioned above as colors for certain  vertices in $W'$. 

We set $k'=4k$ as our new parameter (the number of variables to be set to true).

Similar to the proof of Lemma \ref{lem:vc_membership}, for each $w\in W\cup W'$ and $c\in L(w)$, we create a variable $x_{w,c}$, and give a formula $F^1(X)$, that ensures that each vertex has at least one color:
\[
F^1(X) = \bigwedge_{w\in W\cup W'} \bigvee_{c\in L(w)} x_{w,c}.
\]
The intuition is that $x_{w,c}$ is true, if and only if we color $w$ with $c$. 
Note that when we set exactly $k'$ variables to true and $F^1(X)$ holds, then for each $w\in W\cup W'$, there is exactly one
$c\in L(w)$ with $x_{w,c}$ true. We call this coloring of the vertices in $W\cup W'$ the \emph{coloring given by $X$}.

In the second step, we verify that the coloring given by $X$ does not create a color conflict between vertices in $W\cup W'$. We need to rewrite the formula slightly to avoid the last conjunction. We set
\[
F^2(X) = \bigwedge_{w,w'\in W\cup W', w\neq w'} ~~ \bigwedge_{(c,c')\in (L(w)\times L(w')) \setminus C(w,w')} \neg x_{w,c} \vee \neg x_{w',c'}.
\]
Again, assuming $F^1(X)$ also holds, $F^2(X)$ holds if and only if the coloring given by $X$ does not create
a conflict between vertices in $W\cup W'$. Indeed, the formula verifies that no pair appears that is not allowed.

We next define the color lists $L(w)$ for $w\in W'$.
\begin{itemize}
    \item For all $i\in [k]$,  $L(w_i')=\mathcal{C}$. The interpretation of the color $c$ for $w_i'$ is that, if the colors $\alpha(w_1),\dots,\alpha (w_k)$ were sorted, the $i$th color equals $c$.
    \item For all $i\in [k]$, we set
    $L(w_{k+i}')=[k]\times \mathcal{C}$. The interpretation of the color $(j,c)$ for $w_{k+i}'$ is that $w_i'$ copied color $c$ from $w_j$. By enforcing that the first parts of the colors of $w_{k+i}'$ and $w_{k+i'}'$ are different for distinct $i,i'$, we can enforce that  $\{\alpha(w_1),\dots,\alpha(w_k)\}=\{\alpha(w_1'),\dots,\alpha(w_k')\}$. 
    (We could also have added this to the description of $w_i'$, but separated this out for legibility purposes.)
    \item For $i\in [k-1]$, we set
    $L(w_{2k+i}')=\{(a,b)\in \C\times \C:a\leq b\}$ and we set
    $L(w_{3k}')=\{(b,a)\in \C\times \C:a\leq b\}$. These vertices will store the interval endpoints. 
\end{itemize}
Next, we define the constraint sets $C(w,w')$ for distinct $w,w'\in W\cup W'$, which will be enforced by $F^2$.
\begin{itemize}
     \item For $w,w'\in W$ with $ww'\in E$, we set $C(w,w')=\{(a,b)\in L(w)\times L(w'):a\neq b\}$. So if $F^1$ and $F^2$ holds, the coloring given by $X$ restricted to $W$ is a proper list coloring.    
    \item For $i,i'\in [k]$, we add constraints that ensures each $w_{k+i'}$ copies from the $w_j$ that it claims to:
    \[
    C(w_i,w_{k+i'}')= \{(a,(j,b))\in L(w_i)\times L(w_{k+i'}'):j\neq i \text{ or }a=b\}.
    \]
    We ensure that each $w_{k+i'}'$ `copies' from a different $w_j$ by adding the following constraints for distinct $i,i'\in [k]$:
    \[
    C(w_{k+i}',w_{k+i'}')=\{((j,a),(j',a'))\in L(w_{k+i}')\times L(w_{k+i'}'):j\neq j'\}.
    \]
    \item For $i\in [k]$, we add constraints to ensure the color of $w_i'$ is the same as that of $w_{k+i}'$:
    \[
    C(w_i',w_{k+i}')= \{(a,(j,a))\in L(w_i')\times L(w_{k+i}')\}.
    \] 
    \item For $i\in [k-1]$, we add constraints to ensure the interval endpoints are copied:
    \[
    C(w_i',w_{2k+i}')=\{(a,(a,b))\in L(w_i')\times L(w_{2k+i}')\}
    \]
    and 
    \[
    C(w_{i+1}',w_{2k+i}')=\{(b,(a,b))\in L(w_{i+1}')\times L(w_{2k+i}')\}.
    \]
    Similarly, the color of $w_{3k}'$ has to be a tuple containing the color of $w_k'$ and $w_1'$:
    \[
    C(w_1',w_{3k}')=\{(b,(a,b))\in L(w_1')\times L(w_{3k}')\},
    \]
    \[
    C(w_k',w_{3k}')=\{(a,(a,b))\in L(w_k')\times L(w_{3k}')\}.
    \]
    \item For all $w,w'\in W\cup W'$ for which $C(w,w')$ is not yet defined, we set $C(w,w')=L(w)\times L(w')$.
\end{itemize}
Given $(a,b)\in \mathcal{C}\times \C$, we define $\C_{(a,b)}=\{c\in \C: a<c<b\}$ if $a\leq b$ and $\C_{(a,b)}=\{c\in \C: c<a \text{ or }c>b\}$ if $a>b$ (both sets could be empty). 
The most important observation is the following.
\begin{claim}
If $F^1(X)$ and $F^2(X)$ hold, then the coloring $\alpha$ given by $X$ has the following property. A color  $c\in \mathcal{C}\setminus \{\alpha(w_1),\dots,\alpha(w_k)\}$ if and only if $c\in \C_{(a,b)}$ where $(a,b)=\alpha(w_{2k+i}')$ for some $i\in [k]$.
\end{claim}
Now, we can check whether a vertex $v\in V\setminus W$ can be assigned a color with the following formula:
\[
F^3(X,v) = \bigvee_{c\in L(v)} ~ \bigvee_{i\in[k]}~\bigvee_{(a,b)\in \mathcal{C}\times \mathcal{C}: c\in \mathcal{C}_{(a,b)}} x_{w_{2k+i}',(a,b)}.
\]
Finally, we set
\[
F(X)=F^1(X)\wedge F^2(X) \wedge \bigwedge_{v\in V\setminus W} F^3(X,v).
\]
Note that $F$ is 2-normalized and has 
polynomial size. 
From the discussion above, we see that $F(X)$ holds, if and only if $X$ gives a coloring of the vertices in $W$ that
does not create a conflict between vertices in $W$, and each vertex in $V\setminus W$ can choose a color that does appear among any of the vertices in $W$. Similarly, any valid list coloring of $G$ can be used to give a satisfying assignment of $F(X)$ in which $4k$ variables are set to true. 

We next show how to handle the case when some vertices in $V\setminus W$ are not adjacent to all of $W$.
We may partition the vertices in $V\setminus W$ into  $p\leq 2^k$ vertex sets $I_1,\dots,I_p$ according to their neighbourhood in $W$. Let $W_1,\dots,W_p$ be their corresponding neighbourhoods in $W$.
We create new vertex sets $W_1,\dots,W_p$, making copies of the corresponding vertices in $W$, and set $k'= k+4\sum_{i=1}^p |W_i| $. 

For all $j\in [p]$, the graph induced on $W_j\cup I_j$ is a graph with a vertex cover of size $|W_j|$, such that each vertex outside of the vertex cover is fully adjacent to the vertex cover. We can hence create an auxiliary set $W_j'$ and variables $x^j_{w,c}$ for $w\in W_j\cup W_j'$ and a formula $G^j(X^j)$ as described previously, such that if $4|W_j|$ variables among $X^j$ are set to true, the coloring given on $W_j$ by $X^j$ can be extended to $I_j$ if and only if $G^j(X^j)$ is true. We also add variables $x_{w,c}$ for all $w\in W$ and $c\in L(w)$ and may use a formula $G(X)$ similar to $F^1$ and $F^2$ defined above to ensure each vertex of $W$ receives exactly one color, and the coloring induced on $W$ is proper. Finally, for each $j\in [p]$ we ensure each vertex of $W_j$ `copies' the color from its twin in $W$ using a formula of the form
\[
H^j(X^j,X)=\bigwedge_{w\in W_j} \bigvee_{c,c'\in \mathcal{C}:c\neq c'} (\neg x^j_{w,c}  \vee \neg x_{w,c'}).
\]
Our final formula is
\[
H(X,X^1,\dots,X^p)=G(X) \wedge \bigwedge_{j\in [p]}G^j(X^j) \wedge H^j(X^j,X),
\]
which is again a conjunction of disjunctions as desired. Note that for each $j\in [p]$, exactly $4|W_j|$ variables among $X^j$ will be set to true, since the formula enforces that $k'$ vertices must receive a color (i.e. have a corresponding variable that is set to true).
\end{proof}

\begin{corollary}
\label{cor:listcol_even_membership}
For $d\geq 2$,
    \textsc{List Coloring} with the size of a vertex modulator to treedepth-$d$ is in $\W[2d]$.
\end{corollary}
\begin{proof}[Sketch]
Assume an instance of \textsc{List Coloring} with a graph $G=(V,E)$, lists $L(v)$ for all $v\in V$, modulator set of vertices
$W$ of size $k$, and a rooted forest embedding $T=(V\setminus W, F)$ of depth $d$.
Write $W= \{w_1, \ldots, w_k\}$.

As in the proof of  Lemma \ref{lemma:in}, we define the formula $F^6(X,v,f)$ for all vertices $v$ at depth at most $d-2$ in $T$. 
If $v$ is at depth $d-1$, then we set $F^6(X,v,f)=F^5(X,v,f)\wedge \bigwedge_{u\in B_v} G(X,u,f)$, where $B_v$ are the children of $v$ and $G$ will be defined shortly as a conjunction of disjunctions. In the original definition of $F^6$, it could happen here that $G$ is taken to be a conjunction of disjunctions of conjunctions, so this is where we gain one level of the hierarchy. Note that the definitions of $F^5(X,f,v)$ and the original $F^6(X,f,v)$ ensure that the color assigned by $f$ to a vertex $v$ at depth at most $d-1$, is different from those assigned to its neighbours in $W$ by $X$ and from the color assigned to its adjacent ancestors by $f$.

Compared to the proof Lemma \ref{lemma:in}, the set $X$ contains additional variables of the form $x_{w',c}$ which model the assignment of colors to auxiliary vertices $w'$ in a set $W'$ of size at most $4\cdot 2^k$. We set $k'=k+|W'|$ as the number of variables that may be set to true. We create a formula $H(X)$ that replaces the purpose of the original $F^1$ and $F^2$ from the proof of Lemma \ref{lemma:in}, using rather the $F^1$ and $F^2$ from the proof of Lemma \ref{lem:listcol_vc_membership} (and the additional formulas required to handle the case in which vertices may have different neighbourhoods among $W$).  
We define $G(X,u,f)$ in a similar fashion to the proof of Lemma \ref{lem:listcol_vc_membership}: using the additional variables, we can express that $u$ receives a color that is different from the colors of its neighbours in $W$ using a conjunction of disjunctions. In order to ensure that the color of $u$ is different from the color of its ancestors, we replace every occurrence of $L(u)$ in this formula by the list obtained from $L(u)$ by removing the colors assigned by $f$ to the ancestors of $u$. We combine the formulas in a conjunction, resulting in a formula that is a conjuction of disjunctions of conjunctions \dots of disjunctions, with  $d$ levels of conjunctions and $d$ levels disjunctions.
\end{proof}
The result immediately implies that for $d\geq 2$,
    \textsc{List Coloring} with the size of a vertex modulator to depth-$d$ forest as parameter is also in $\W[2d]$.
    
We note that in the proof of Corollary \ref{cor:listcol_even_membership}, we could still allow almost all edges to be arbitrary constraints, except those between a vertex at depth $d$ and a vertex in $W$. 
In particular, a variation on $\csp$ parameterized by vertex modulator to treedepth-$d$ graph, is in $\W[2d]$ if we place the restriction that constraints between depth-$d$ vertices $v$ and vertices $w$ in the modulator, take the form of a list coloring constraint:
\[
C(v,w)=\{(c,c')\in D(v)\times D(w):c\neq c'\}.
\] 
In order to obtain a hardness result, one could mimic the reduction done in the proof Lemma \ref{lemma:w2d+1hard}, where vertices at depth $d$ now only need to verify a disjunction, rather than a disjunction of conjunctions. 
We may further assume no negations appear in this disjunction by the following result.
\begin{theorem}[Downey and Fellows, see \cite{DowneyF99,DowneyF13}]
\label{theorem:df2even}
Let $t\geq 4$ be even. \textsc{Weighted Monotone $t$-Normalized Satisfiability} is complete for $\W[t]$.
\end{theorem}
We recall that a $t$-normalized expression is said to be monotone if each literal is a variable (without negation).

Suppose that a vertex $v$ at depth $d$ wishes to verify a disjunction $F_v=\bigvee_{x\in S_v}x$. We could give $v$ the domain $D(v)=[k]$ and place the constraint that $C(v,w_i)=\{(j,c)\in D(v)\times D(w_i): i\neq j \text{ or }c\in S_v \}$ for all $i\in [k]$.  The value $j$ is assigned to $v$ if $w_j$ gets assigned a variable that appears in its disjunction.
However, only restricting the size of the list of $v$ does not appear to be enough to prove $\W[2d]$-membership; for this we also seem to require some form of restriction on the constraints. For example, $\W[2d]$-membership can be shown when for depth-$d$ vertices we restrict the lists and the constraints to $W$ vertices to be exactly as described above, where $S_v\subseteq \mathcal{C}$ can be chosen arbitrarily. We leave it as an interesting open problem to find a more natural $\W[2d]$-complete problem than this.

\section{XSLP and treedepth}
\label{sec:xslp}

In this section we discuss the class $\XSLP$ and its various characterizations. As discussed in Section~\ref{sec:intro}, we actually define two variants of this class, depending on the kind of reductions that we would like to speak about. Let $\sfrac{\csp}{\td}$ denote the following parameterized problem. We are given a $\csp$ instance $I$ and an elimination forest of the Gaifman graph of $I$ of depth at most $k$, which is the parameter. The task is to decide whether $I$ is satisfiable. Then the two variants of $\XSLP$ are defined as the closures of this problem under $\pl$- and $\fpt$-reductions, respectively:
$$\XSLP=\left[\sfrac{\csp}{\td}\right]^{\pl}\qquad\textrm{and}\qquad \XSLP^+=\left[\sfrac{\csp}{\td}\right]^{\fpt}.$$
That is, $\XSLP$ consists of all parameterized problems that are $\pl$-reducible to $\sfrac{\csp}{\td}$, and $\XSLP^+$ is defined similarly, but with $\fpt$-reductions in mind. 

Note that in the $\sfrac{\csp}{\td}$ problem we assume that a suitable elimination forest is provided on input. This is to abstract away the need of computing such an elimination forest; the complexity of this task is also an interesting question, but lies beyond the scope of this work.

\subsection{A machine characterization}\label{sec:machine-char}

We first give a machine characterization of $\XSLP$. We will use a model of an {\em{alternating read-once stack machine}}, or {\em{AROSM}} for brevity, which we now define. We assume familiarity with standard Turing machines, on which we build our model.

An alternating read-once stack machine $M$ is a Turing machine that has access to three types of memory, each using $\{0,1\}$ as the alphabet:
\begin{itemize}
\item a read-only input tape;
\item a working tape; and
\item a read-once stack.
\end{itemize}
The input tape and the working tape are accessed and manipulated as usual, by a head that may move, read, and (in the case of the working tape) write on the tape.  The input to the machine is provided on the input tape. On the other hand, the stack is initially empty and the machine may, upon any transition, push a single symbol onto the stack. It cannot, however, read the stack until the final step of the computation. More precisely, the acceptance condition is as follows: The machine has a specified {\em{final state}}. Once it is reached, the computation finishes and the machine reads the $i$th bit of the stack, where $i$ is the number whose binary encoding is the current content of the working tape. If this bit is $1$, then $M$ accepts, and otherwise it rejects. 

A {\em{configuration}} of $M$ is a $5$-tuple consisting of the state, the content of the working tape, the content of the stack, and the positions of the heads on the input and the working tape.

Further, $M$ is an {\em{alternating machine}}, which means that its states are partitioned into three types: {\em{existential states}}, {\em{universal states}}, and {\em{deterministic states}}. A configuration of a machine is existential/universal/deterministic if its state is so.
When the state of the machine is deterministic, there is exactly one transition allowed. At existential and universal states, there are always two transitions allowed; these will be named the {\em{$0$-transition}} and the {\em{$1$-transition}}. The acceptance is defined as usual in alternating machines: when in an existential state, $M$ may accept if at least one allowed transition leads to a configuration from which it may accept, and in a universal state we require that both transitions lead to configurations from which $M$ may accept. The notion of a machine deciding a (parameterized) problem is as usual.

The {\em{$\forall$ computation tree}} of $M$ for input $x$ is defined as a tree of configurations with the following properties:
\begin{itemize}
\item the root is the initial configuration with input $x$;
\item the leaves are configurations with the final state;
\item every deterministic and every existential configuration has exactly one child, which is the unique, respectively any of the two configurations to which the machine may transit;
\item every universal configuration has exactly two children, corresponding to the two configurations to which the machine may transit.
\end{itemize}
It follows that $M$ accepts input $x$ if there is a $\forall$ computation tree for input $x$ where every leaf is a configuration in which $M$ accepts. We call such $\forall$ computation trees {\em{accepting}}.

A {\em{branch}} of a (rooted) tree is a root-to-leaf path.
For a $\forall$ computation tree $T$ of machine~$M$, we define the following quantities:
\begin{itemize}
\item The {\em{working space}} of $T$ is the minimum number $i$ such among configurations present in~$T$, the head on the working tape is never beyond the $i$th cell.
\item The {\em{stack size}} of $T$ is the maximum size of the stack among all configurations in $T$.
\item The {\em{nondeterminism}} of $T$ is the maximum number of existential configurations on any branch of $T$.
\item The {\em{conondeterminism}} of $T$ is the maximum number of universal configurations on any branch of $T$.
\item The {\em{alternation}} of a branch of $T$ is the minimum number of blocks into which the branch can be partitioned so that each of the blocks does not simultaneously contain an existential and a universal configuration. The {\em{alternation}} of $T$ is the maximum alternation on any branch of $T$.
\end{itemize}
We say that a machine $M$ {\em{decides}} a parameterized problem $Q$ using certain resources among those described above, if for any input $(x,k)$, we have $(x,k)\in Q$ if and only if there is an accepting $\forall$ computation tree for $(x,k)$ that has the resources bounded as prescribed.

Having all the necessary definitions in place, we can state the main result of this section.

\begin{restatable}{theorem}{machinechar}
\label{thm:machine-characterization}
The following conditions are equivalent for a parameterized problem $Q$.
\begin{enumerate}[(1)]
\item\label{a:xslp} $Q\in \XSLP$;
\item\label{a:machine} $Q$ can be decided by an alternating read-once stack machine that for input $(x,k)$ with $|x|=n$, uses working space at most $f(k)+\Oh(\log n)$, stack size $f(k)\log n$, nondeterminism  $f(k)\log n$, co-nondeterminism $f(k)+\Oh(\log n)$, and alternation $f(k)$, for some computable function $f$.
\end{enumerate}
\end{restatable}

Before we proceed to the proof of Theorem~\ref{thm:machine-characterization}, let us discuss the necessity of different resource restrictions described in \eqref{a:machine}:
\begin{itemize}
\item Increasing the working space to $f(k)\log n$ (and thus rendering the stack, the non-determinism and the co-nondeterminism unnecessary) would make the machine model at least as powerful (and in fact, equivalently powerful) as deterministic Turing machines with $f(k)\log n$ space; this corresponds to a class called~$\mathrm{XL}$. As $\mathrm{XL}^+$ (the closure of $\mathrm{XL}$ under fpt reductions) contains $\mathrm{AW}[\mathrm{SAT}]$~\cite[Exercise~8.39]{FlumGrohe}, the supposition that the amended model is still equivalent to $\XSLP$ would imply the inclusion $\mathrm{AW}[\star]\subseteq \mathrm{AW}[\mathrm{SAT}]\subseteq \XSLP^+$. From the logic characterization that will be provided in Section~\ref{sec:logic-char} it follows that $\mathrm{AW}[\star]\supseteq \XSLP^+$, so in fact we would obtain a collapse $\mathrm{AW}[\star]=\mathrm{AW}[\mathrm{SAT}]=\XSLP^+$.
    \item If we increase the bound on allowed co-nondeterminism to $f(k) \log n$, thus matching the bound on the allowed nondeterminism, then it is not hard to see that the obtained machine model would be able to solve the model-checking problem for first-order logic on general relational structures, which is $\mathrm{AW}[\star]$-complete. Consequently, if the amended machine model was still equivalent to $\XSLP$, we would again obtain equality $\mathrm{AW}[\star]=\XSLP^+$, which we consider unlikely.
    \item If we let the machine use unbounded nondeterminism, then already for $k$ constant and assuming no use of co-nondeterminism, our machines would be able to solve every problem in $\mathrm{NL}$, including {\sc{Directed Reachability}}. If the obtained machine model was still equivalent to $\XSLP$, then {\sc{Directed Reachability}} would be reducible (in $\mathrm{L}$) to $\csp$ on graphs of constant treedepth. But the latter problem is actually in $\mathrm{L}$, so we would obtain $\mathrm{L}=\mathrm{NL}$.
    \item We believe that increasing the alternation from $f(k)$ to $f(k)+\Oh(\log n)$ yields a strictly more powerful machine model, though at this point we cannot pinpoint any concrete collapse that would be implied by the converse. However, it is not hard to check that an AROSM with resource bounds as in Theorem~\ref{thm:machine-characterization}, but alternation $f(k)+\Oh(\log n)$, is able to solve $\csp$ instances with Gaifman graphs of treedepth as large as $\log n$, but with all domains of size at most $k$. We do not see how to reduce this problem to $\csp$ with domains of unbounded size, but treedepth bounded by $f(k)$.
    \item It is an interesting question whether the $f(k)\log n$ bound on the stack size can be lifted; that is, whether allowing unbounded stack size strictly increases the power of the considered machine model. On one hand, in all our proofs, the stack is essentially only used to store nondeterministic bits, and in any run there are at most $f(k)\log n$ of them anyway. So if the stack is used only for this purpose, then it is immaterial whether its size is bounded by $f(k)\log n$ or unbounded. On the other hand, the restriction on the stack size plays an important role in the proof of the implication \eqref{a:machine} $\Rightarrow$ \eqref{a:xslp} of Theorem~\ref{thm:machine-characterization}. We leave resolving this question open. 
\end{itemize}

The remainder of this section is devoted to the proof of Theorem~\ref{thm:machine-characterization}. Naturally, the argument is split into the forward and the backward implication.

\subsubsection{From $\sfrac{\csp}{\td}$ to AROSM acceptance}

We first prove the simpler implication \eqref{a:xslp} $\Rightarrow$ \eqref{a:machine} of Theorem~\ref{thm:machine-characterization}.
 Since any $\pl$-reduction can be computed by an AROSM with working space bounded by $f(k)+\Oh(\log n)$, it suffices to prove the following statement: $\sfrac{\csp}{\td}$ problem can be decided using an AROSM $M$ that works under the resource bounds stipulated in Theorem~\ref{thm:machine-characterization}. Here, $k$ is the depth of the given elimination forest of the Gaifman graph of the input instance. 
Let then $I=(G,\{D(u)\colon u\in V(G)\},\{C(u,v)\colon uv\in E(G)\},T)$ be the given instance of $\csp$, where $T$ has depth $k$. By adding a dummy variable connected by trivial constraints to all the other variables (which increases $k$ by $1$), we may assume that $G$ is connected and $T$ is an elimination tree of $G$. By Lemma~\ref{lem:conon}, we may assume that $T$ is equipped with a labelling $\lambda\colon E(T)\to \{0,1\}^\star$ satisfying the conclusion of the lemma.

The idea is as follows: The constructed AROSM $M$ guesses a satisfying assignment to $I$ in a top-down manner on $T$, using nondeterminism to guess the evaluation of every next variable $u$, and conondeterminism to verify whether the currently guessed partial evaluation can be extended to all the subtrees rooted at the children of $u$. More precisely, at every point, the configuration of $M$ consists of the following:
\begin{itemize}
\item The working tape contains the vertex $u$ of $G$ that is currently processed by $M$. This vertex is initially the root of $T$.
\item The stack contains a partial evaluation $\eta$ that maps every strict ancestor $w$ of $u$ to a value $\eta(w)\in D(w)$.
\end{itemize}
Thus, the content of the working tape consists of $\Oh(\log n)$ bits and the stack contains $\Oh(k\log n)$ bits. When presented with such a configuration, machine $M$ does the following:
\begin{itemize}
\item Nondeterministically guess $\eta(u)\in D(u)$ and push $\eta(u)$ onto the stack. (We assume here that the encoding of $\eta(u)$ takes exactly $\lceil \log n\rceil$ bits.)
\item If $u$ has children, conondeterministically guess a label $\lambda(uv)$ of an edge $uv$ connecting $u$ with a child $v$. Noting that $v$ can be uniquely determined from the pair $(u,\lambda(uv))$, compute $v$, change the currently processed vertex to $v$, and continue.
\item If $u$ has no children, proceed to the verification phase, which we describe in a moment.
\end{itemize}
Observe that guessing $\eta(u)$ requires $\log n$ nondeterministic bits, while guessing $\lambda(uv)$ can be done using $2|\lambda(uv)|+1$ conondeterministic bits, as explained in the footnote on page~\pageref{pg:ft}.  Since $T$ has depth at most $k$, by the properties of $\lambda$ it follows that on every branch $M$ uses $\Oh(k\log n)$ nondeterminism, $\Oh(k+\log n)$ conondeterminism, and alternation at most $2k$. Moreover, the content of the working tape always consists of $\Oh(\log n)$ bits, while the stack never grows larger than $\Oh(k\log n)$.

We now describe the verification phase, which is triggered once the machine has guessed a partial assignment $\eta$ mapping every ancestor $w$ of a leaf $\ell$ (including $\ell$ itself) to a value $\eta(w)\in D(w)$. It remains to verify whether $\eta$ satisfies all binary constraints present in $I$ that bind two ancestors of $\ell$; that is, whether
$$(\eta(w),\eta(w'))\in C(w,w')$$
for every edge $ww'$ where $w$ and $w'$ are ancestors of $\ell$. To do this, the machine does the following:
\begin{itemize}
\item conondeterministically guess the edge $ww'$;
\item nondeterministically guess two values $a\in D(w)$ and $a'\in D(w')$ and verify that $(a,a')\in C(w,w')$;
\item verify that $a=\eta(u)$ and $a'=\eta(u')$ by conondeterministically guessing an index of a bit from the stack to be verified, reading this one bit, and checking it against the corresponding bit in $a$ or $a'$.\footnote{Formally, we assumed that the machine accepts if the read bit is $1$. The verification described above can be simulated within this mechanism by encoding every $0$ as $01$ and every $1$ as $10$ on the stack, and reading the right bit if verifying against a $0$, or the left bit if verifying against a $1$.}
\end{itemize}
Thus, the verification phase uses alternation $3$, $\Oh(\log n)$ bits of nondeterminism, and $\Oh(\log n)$ bits of conondeterminism.

It is straightforward to see that the described machine $M$ accepts the input instance $I$ if and only if $I$ has a satisfying assignment. Since the bounds on used resources are as described in Theorem~\ref{thm:machine-characterization}, the implication \eqref{a:xslp} $\Rightarrow$ \eqref{a:machine} follows.

\subsubsection{From AROSM acceptance to $\sfrac{\csp}{\td}$}

We now proceed to the more difficult implication \eqref{a:machine} $\Rightarrow$ \eqref{a:xslp} of Theorem~\ref{thm:machine-characterization}. The main idea is that we introduce a restricted variant of a {\em{regular}} AROSM, which is an AROSM whose $\forall$ computation tree has a very specific shape, computable from $k$ and the length of the input. We will then show two lemmas: (i) for every AROSM there is an equivalent regular one, and (ii) acceptance of a regular AROSM can be reduced to $\sfrac{\csp}{\td}$. The main point in this strategy is that the assumption that the computation tree is fixed allows us to fix it as the elimination tree of the Gaifman graph of the constructed $\csp$ instance.

More precisely, we will be working with the {\em{contracted $\forall$ computation trees}} defined as follows. Let $T$ be a $\forall$ computation tree of an AROSM $M$, where without loss of generality we assume that the starting state of $M$ is universal. A {\em{universal block}} of $T$ is an inclusion-wise maximal subtree $A$ of $T$ such that the root of $A$ is a universal configuration and $A$ does not contain existential configurations. Note that removing all universal blocks from $T$ breaks $T$ into a collection of disjoint paths consisting only of deterministic and existential configurations; these will be called {\em{existential blocks}}. The {\em{contraction}} of $T$ is the tree $T'$ whose nodes are universal blocks of $T$, where the ancestor order is naturally inherited from~$T$: one block is an ancestor of the other in $T'$ if this holds for their roots in $T$. Note that a universal block $B$ is a child of a universal block $A$ in $T'$ if and only if there is an existential block $C$ that connects the root of $B$ with a leaf of $A$. Thus, the edges of $T'$ are in one-to-one correspondence with the existential blocks of $T$.

\begin{definition}
An AROSM $M$ is {\em{regular}} if given $(1^n,k)$ one can in parameterized logspace compute a rooted tree $T_{n,k}$ with the following properties:
\begin{itemize}
    \item $T_{n,k}$ has depth at most $f(k)$, for some computable function $f$; and 
    \item for any input $(x,k)$ with $|x|=n$, if $M$ accepts $(x,k)$, then $M$ has a $\forall$ computation tree accepting $(x,k)$ whose contraction is $T_{n,k}$.
\end{itemize} 
\end{definition}

With this definition in place, we can state the two lemmas described before.

\begin{lemma}\label{lem:regularization}
If a parameterized problem $Q$ can be decided by an AROSM $M$ using the resource bounds stated in Theorem~\ref{thm:machine-characterization}, then it can also be decided by a regular AROSM $M'$ using such resource bounds.
\end{lemma}

\begin{restatable}{lemma}{regularreduction}
\label{lem:regular-reduction}
If $Q$ can be decided by a regular AROSM $M$ using the resource bounds stated in Theorem~\ref{thm:machine-characterization}, then $Q\in \XSLP$.
\end{restatable}
The \eqref{a:machine} $\Rightarrow$ \eqref{a:xslp} implication of Theorem~\ref{thm:machine-characterization} follows directly by combining the two lemmas above. The proof of Lemma~\ref{lem:regular-reduction} is a conceptually straightforward, though technically a bit involved encoding of a $\forall$ computation tree of the machine through an instance of $\csp$ whose elimination tree is (roughly) $T_{n,k}$. We give this proof in Section~\ref{sec:regular-csp-arXiv}. The proof of Lemma~\ref{lem:regularization} is the interesting part of the argument, as it involves the notion of {\em{universal trees}}. 

Before we proceed, let us state a simple lemma that is used in our proofs several times. We include the proof for completeness.

\begin{restatable}{lemma}{conon}\label{lem:conon}
Suppose $T$ is a rooted tree with $N$ leaves. Then there exists a labelling $\lambda$ that maps every edge $e$ of $T$ to a binary string $\lambda(e)\in \{0,1\}^\star$ with the following properties:
\begin{itemize}
\item For every node $u$, the labels of edges connecting $u$ with its children are pairwise different.
\item For every leaf $\ell$, the total length of labels on the root-to-$\ell$ path in $T$ is at most $\lceil \log N\rceil$.
\end{itemize}
Moreover, given $T$ the labelling $\lambda$ can be computed in deterministic logarithmic space.
\end{restatable}
\begin{proof}
Let $p=\lceil \log N\rceil$ and let $\preceq$ be the preorder on $T$ restricted to the leaves.
Enumerate the leaves of $T$ in the order $\preceq$ using indices $0,1,\ldots,N-1$. For a leaf~$\ell$, let $\alpha(\ell)\in \{0,1\}^{p}$ be the length-$p$ binary encoding of the index associated with $\ell$. For a node~$u$, let $\mathsf{left}(u)$ and $\mathsf{right}(u)$ be the $\preceq$-minimal and $\preceq$-maximal leaf descendant of $u$. We define
$$\eta(u)=\textrm{longest common prefix of }\mathsf{left}(u)\textrm{ and }\mathsf{right}(u).$$
Clearly, if $v$ is a child of $u$, then $\eta(u)$ is a prefix of $\eta(v)$. So we may define the labelling $\lambda$ as follows: $\lambda(uv)$ is the string $w$ such that $\eta(v)=\eta(u)w$.

Note that if $\ell$ is a leaf of $T$ and $(e_1,\ldots,e_m)$ is the sequence of edges on the root-to-$\ell$ path, then $$\alpha(\ell)=\lambda(e_1)\lambda(e_2)\ldots\lambda(e_m).$$
Since $|\alpha(\ell)|\leq p$, the second property asserted in the lemma statement follows.

For the first property, suppose for contradiction that there exists a node $u$ of $T$ and its two children $v$ and $v'$ such that $\lambda(uv)=\lambda(uv')$; call this label $w$. Let $s$ be the concatenation labels along the root-to-$u$ path. Thus, $sw$ is a prefix of all the following four strings: $\alpha(\mathsf{left}(v))$, $\alpha(\mathsf{right}(v))$, $\alpha(\mathsf{left}(v'))$, and $\alpha(\mathsf{right}(v'))$. That $v\neq v'$ even though $\lambda(uv)=\lambda(uv')$ implies that both $v$ and $v'$ are not leaves. By the definition of $\lambda$ it follows that
\begin{itemize}
\item $sw0$ is a prefix of $\alpha(\mathsf{left}(v))$ and of $\alpha(\mathsf{left}(v'))$; and
\item $sw1$ is a prefix of $\alpha(\mathsf{right}(v))$ and of $\alpha(\mathsf{right}(v'))$.
\end{itemize}
This means that in $\preceq$, $\mathsf{left}(v)$ and $\mathsf{left}(v')$ are both smaller than $\alpha(\mathsf{right}(v))$ and $\alpha(\mathsf{right}(v'))$. This is a contradiction with the choice of $\preceq$ as the preorder in $T$ restricted to the leaves.
\end{proof}

\subsubsection{Regularization}
\label{sec:reg}
We now prove Lemma~\ref{lem:regularization}. We need the following definitions. An {\em{ordered tree}} is a rooted tree where for every vertex $u$ there is a linear order $\preceq$ on the children of $u$. An {\em{embedding}} of an ordered tree $S$ into an ordered tree $T$ is an injective mapping $\phi\colon V(S)\to V(T)$ such that
\begin{itemize}
\item the root of $S$ is mapped to the root of $T$, and
\item for every node $u$ of $S$, the children of $u$ in $S$ are mapped to distinct children of $\phi(u)$ in $T$ in an order-preserving way: if $v\prec v'$ are distinct children of $u$ in $S$, then $\phi(v)\prec \phi(v')$.
\end{itemize}
We will use the following result about the existence of {\em{universal trees}}.

\begin{lemma}[follows from Jurdziński and Lazić~\cite{JurdzinskiL17}, see also Theorem 2.2 of~\cite{CzerwinskiDFJLP19}]\label{lem:universal-tree}
For every pair of integers $n,k\in \N$ there exists an ordered tree $U_{n,k}$ such that 
\begin{itemize}
\item $U_{n,k}$ has depth $k$;
\item $U_{n,k}$ has at most $2n\cdot \binom{\lceil \log n\rceil +k+1}{k}$ leaves; and
\item for every ordered tree $T$ of depth at most $k$ and with at most $n$ leaves, there is an embedding of $T$ into $U_{n,k}$.
\end{itemize}
Moreover, given $(1^n,k)$, the tree $U_{n,k}$ can be computed parameterized logspace.
\end{lemma}

We remark that the claim about the computability of $U_{n,k}$ in parameterized logspace  is not present in~\cite{CzerwinskiDFJLP19,JurdzinskiL17}, but follows directly from the recursive construction presented there. In fact, we will also need the following property, which again follows directly from the construction, and which strengthens the embedding property stated in Lemma~\ref{lem:universal-tree}.

\begin{lemma}\label{lem:stronger-embedding}
For every node $u$ of $U_{n,k}$, the subtree of $U_{n,k}$ rooted at $u$ is isomorphic to $U_{n',k'}$ for some $n'\leq n$ and $k'\leq k$; the labeling of nodes of $U_{n,k}$ with suitable numbers $n',k'$ can be computed along with $U_{n,k}$ within the algorithm of Lemma~\ref{lem:universal-tree}. Moreover, if $n_1,\ldots,n_p$ are nonnegative integers such that $n_1+\ldots+n_p\leq n$, then there are distinct children $v_1\prec v_2\prec \ldots \prec v_p$ of the root of $U_{n,k}$ such that for every $i\in \{1,\ldots,p\}$, the subtree of $U_{n,k}$ rooted at $v_i$ is isomorphic to $U_{n'_i,k-1}$ for some $n'_i\geq n_i$.
\end{lemma}

Finally, observe that
$$2n\cdot \binom{\lceil \log n\rceil +k+1}{k}\leq 2n\cdot 2^{\lceil \log n\rceil +k+1}\leq \Oh(2^k\cdot n^2),$$
hence $U_{n,k}$ has $\Oh(2^k\cdot n^2)$ leaves.

We proceed to the proof of Lemma~\ref{lem:regularization}. Let us fix an AROSM $M$ that on any input $(x,k)$ with $|x|\leq n$, uses $f(k)\log n$ nondeterminism, $f(k)+d\log n$ conondeterminism, $f(k)$ alternation, $f(k)+d\log n$ working space, and $f(k)\log n$ stack size, where $f$ is a computable function and $d\in \N$ is a constant. We may assume w.l.o.g. that the starting state of $M$ is universal. Denote $K=f(k)$ and $N=2^{f(k)+\lceil d\log n\rceil}\leq 2^{f(k)+1}\cdot n^d$; then $K$ is an upper bound on the depth and $N$ is an upper bound on the total number of leaves of any $\forall$ computation tree accepting $(x,k)$ within the stipulated resources. By Lemma~\ref{lem:universal-tree}, we may compute the universal tree $U_{N,K}$ in deterministic space $h(k)+\Oh(\log n)$ for a computable function $h$. Note that $U_{N,K}$ has $N'=\Oh(2^K\cdot N^2)\leq \Oh(2^{3f(k)}\cdot n^{2d})$ leaves. The tree $U_{N,K}$ will serve the role of $T_{n,k}$ in the proof. Also, we use Lemma~\ref{lem:conon} to compute a suitable labeling $\lambda$ of the edges of $U_{N,K}$ in which the total length of labels on every branch of $U_{N,K}$ is at most $\lceil \log N'\rceil \leq 3f(k)+2d\log n+\Oh(1)$.

We are left with designing an AROSM $M'$ that is equivalent to $M$, in the sense that $M'$ accepts input $(x,k)$ if and only if $M$ does, and in such case the contracted $\forall$ computation tree of $M'$ on $(x,k)$ may be $U_{N,K}$. The idea is that machine $M'$ will simulate $M$ while inserting some dummy computation to ``fill'' the contracted $\forall$ computation tree of $M$ to~$U_{N,K}$. However, we will need to be very careful about how the conondeterminism of $M$ is simulated.

A {\em{stackless configuration}} is a configuration of $M$, but without specifying the content of the stack; that is, it consists of the state of $M$, the content of the working tape, and the positions of the heads on the input and the working tapes.
For a universal stackless configuration $c$ of $M$, we define the {\em{universal block}} rooted at $c$, denoted $U(c)$, as a rooted tree of stackless configurations that is obtained just as the $\forall$ computation tree, except that $M$ starts at $c$ and we do not continue the simulation once the final state or any existential configuration is reached. Note here since $M$ cannot read the stack except for the end of the computation, $U(c)$ is uniquely defined for every stackless configuration $c$.
Thus, the leaves of $U(c)$ are existential or final (stackless) configurations, and whenever $c$ is present in a $\forall$~computation tree $T$ of $M$, $T$ contains a copy of $U(c)$ rooted at $c$ as a subtree. 

The next claim shows that given a stackless configuration $c$, the universal block $U(c)$ can be computed within the allowed resources.

\begin{claim}\label{cl:generate-block}
Given a stackless configuration $c$ of $M$, the universal block $U(c)$, together with a labelling of its edges with transitions taken, can be computed in deterministic space $h(k)+\Oh(\log n)$, for some computable $h$.
\end{claim}
\begin{claimproof}
Let $Z=f(k)+\lfloor d\log n\rfloor$.
Observe that for every binary string $r\in \{0,1\}^{Z}$, we can compute the branch of $U(c)$ that takes the consecutive conondeterministic choices as prescribed by the consecutive bits of $r$. To do this, just simulate $M$ starting from $c$ and, whenever a conondeterministic choice needs to be made, use the next bit of $r$ to determine how it is resolved. (This simulation stops when an existential or a final configuration is encountered.) Having this subprocedure, the whole $U(c)$ can be easily computed by iterating through consecutive strings $r\in \{0,1\}^Z$ and outputting the branches of $U(c)$ one after the other. (Strictly speaking, from every next branch we output only the part after diverging from the previous branch.) Finally, note that  $r$ can be stored within the allowed~space.
\end{claimproof}

With Claim~\ref{cl:generate-block} established, we proceed to the construction of $M'$. For the sake of the proof, suppose $M$ has a $\forall$ computation tree $T$ that is accepting and uses the allowed resources. Machine $M'$ tries to verifies the existence of such $T$ by traversing the universal tree $U_{N,K}$ and guessing, along the way, how the contraction $T'$ of $T$ embeds into $U_{N,K}$. By Lemma~\ref{lem:universal-tree} we know that such an embedding always exists. The traversal of $U_{N,K}$ will be done in such a way that the contracted $\forall$ computation tree of $M'$ will be always $U_{N,K}$.

At every point of computation, $M'$ stores on its working tape a node $u$ of $U_{N,K}$ and its contracted $\forall$ computation tree from this point on should be the subtree of $U_{N,K}$ rooted at $u$. Machine $M'$ is always either in the {\em{real mode}} or in the {\em{dummy mode}}. In the real mode, $M'$ is in the process of guessing a subtree of $T$. Therefore, then $M'$ holds the following data:
\begin{itemize}
\item On the working tape, $M'$ stores a stackless configuration $c$ of $M$. The reader should think of $c$ as of the configuration of $M$ at the root of a universal block of $T$.
\item On its own stack, $M'$ holds the content of the stack of $M$.
\item Additionally on the working tape, $M'$ stores two integers $a$ and $b$, denoting the total number of nondeterministic and conondeterministic bits used by $M$ so far, respectively. (In other words, $a$ and $b$ are the total number of existential and universal configurations visited so far on a branch of $T$.) We maintain the following invariant: the subtree of $U_{N,K}$ rooted at $u$ is $U_{N',K'}$ for some $K'\leq K$ and $N'$ such that $N'\geq N/2^b$.
\end{itemize}
Then the task of $M'$ is to verify the existence of a subtree $S$ of a $\forall$ computation tree of $M$ such that 
\begin{itemize}
\item $S$ has $c$ supplied with the current content of the stack at its root; 
\item $S$ embeds into the subtree of $U_{N,K}$ rooted at $u$;
\item the nondeterminism and the conondeterminism of $S$ together with $a$ and $b$ add to at most $f(k)\log n$ and $f(k)+d\log n$, respectively; and
\item $S$ is accepting, that is, every leaf of $S$ is an accepting configuration.
\end{itemize}
In the dummy mode, $M'$ is not guessing any part of $T$, so its task is to perform some meaningless computation in order to make its contracted $\forall$ computation tree equal to the subtree of $U_{N,K}$ rooted at $u$. So in this mode, $M'$ holds on its working tape only the node $u$. 

We now explain steps taken by $M'$ in the real mode. Given $c$, $M'$ applies the algorithm of Claim~\ref{cl:generate-block} to compute the universal block $U(c)$. (Formally speaking, $U(c)$ is not computed explicitly, as it would not fit within the working space, but at any point a bit from the description of $U(c)$ is needed, we run the algorithm of Claim~\ref{cl:generate-block} to compute this bit.) Let $\ell_1,\ldots,\ell_p$ be the leaves of $U(c)$, in the order as they appear in the description of $U(c)$. Informally, we wish to fit in $U(c)$ into the computation tree of $M'$ while keeping enough `space' for the remaining computations $M$ may wish to perform, without knowing how the computation will continue at the leaves.
For every $i\in \{1,\ldots,p\}$, let $b_i$ be total number of universal configurations on the branch of $U(c)$ finishing at $\ell_i$. By assumption, the subtree of $U_{N,K}$ rooted at $u$ is isomorphic to $U_{N',K'}$ for some $N'\geq N/2^b$ and $K'\leq K$. Similarly, we would like to find children $v_1\prec v_2\prec \ldots \prec v_p$ of $u$ in $U_{N,K}$ such that the subtree rooted at each $v_i$ is isomorphic to $U_{N_i,K'-1}$ where $N_i\geq N/2^{b+b_i}$. This follows from Lemma~\ref{lem:stronger-embedding}: we check that
\[
\sum_{i=1}^p N/2^{b+b_i} \leq  N '\sum_{i=1}^p 2^{-b_i}=N',
\]
where the last equality follows since $U(c)$ is a binary tree. 
Note that given $b_1,\ldots,b_p$, we may compute the corresponding children $v_1,\ldots,v_p$ with sufficiently large subtrees in logarithmic space greedily: having found $v_i$, we can set $v_{i+1}$ to be the $\prec$-smallest child $v$ of $u$ such that $v_{i}\prec v$ and the subtree rooted at $v$ is isomorphic to $U_{N'',K'-1}$ for some $N''\geq m_i$. Hence, from now on we assume that the children $v_1,\ldots,v_p$ are given to us. (Again, formally, when we need any $v_i$, we run the logarithmic space algorithm computing $v_1,\ldots,v_p$ to retrieve the sought $v_i$.)

Machine $M'$ conondeterministically guesses the label $\lambda(uv)$ of an edge $uv$ connecting $u$ with a child $v$; this can be done using conondeterministic $2|\lambda(uv)|+1$ bits\footnote{For instance, the machine can guess consecutive bits of $\lambda(uv)$ interleaved with symbols $0$ and $1$, where $0$ denotes ``continue guessing'' and $1$ denotes ``end of $\lambda(uv)$''.}.\label{pg:ft} Noting that the pair $(u,\lambda(uv))$ uniquely determines $v$, we can now compute $v$. We have two cases:
\begin{itemize}
\item Suppose $v=v_i$ for some $i\in \{1,\ldots,p\}$. Then $M'$ simulates all transitions of $M$ on the path from $c$ to the leaf $\ell_i$ in $U(c)$ (this may include some push operations). If $\ell_i$ is a final configuration, $M'$ finishes the computation and verifies acceptance in the same way $M$ would do. Otherwise, if $\ell_i$ is an existential configuration, $M'$ further nondeterministically simulates $M$ using its own nondeterminism, until a final or a universal configuration is encountered, or the bound of $f(k)\log n$ on the total number of nondeterministic steps is exceeded (together with $a$). In case of a final configuration, we do the same as before: machine $M'$ concludes the computation and verifies whether $M$ accepts. In case of a universal configuration, say $c'$, $M'$ moves the currently considered node of $U_{N,K}$ from $u$ to $v$, and proceeds with working with $c'$ at $v$. The counters $a$ and $b$ are updated by the total number of nondeterministic and conondeterminisitc bits used between $\ell$ and $c'$ and between $c$ and $\ell$, respectively.
Note here that the content of the stack has been appropriately updated while simulating the transitions of $M$ from $c$ to $c'$.
\item Suppose $v\notin \{v_1,\ldots,v_p\}$. Then $M'$ moves the currently considered node of $U_{N,K}$ from $u$ to $v$, but enters $v$ in the dummy mode.
\end{itemize}
This concludes the description of the behavior of $M'$ in the real mode.

Finally, when in the dummy mode, machine $M'$ does as follows:
\begin{itemize}
\item If $u$ is a leaf, $M'$ just accepts.
\item If $u$ is not a leaf, $M'$ conondeterministically chooses a label $\lambda(uv)$ of an edge $uv$ connecting $u$ with a child $v$, using $2|\lambda(u,v)|+1$ conondeterminstic bits. Then $M'$ computes $v$, performs a trivial nondeterministic transition, and enters $v$, again in the dummy mode.
\end{itemize}
This completes the construction of $M'$.

From the construction it follows that the contracted $\forall$ computation tree of $M'$ on $(x,k)$ is always $U_{N,K}$, hence $M'$ is regular. Moreover, on every branch $M'$ uses as much nondeterminism as $M$, that is, at most $f(k)\log n$, while the conondeterminism of $M'$ is bounded by $2\lceil \log N'\rceil +k\leq 6f(k)+4d\log n+k+\Oh(1)$, by the assumed properties of the labeling~$\lambda$. The maximum stack length of $M'$ is the same as that of $M$, while on its working tape, $M'$ holds always at most one configuration of $M$ plus $h(k)+\Oh(\log n)$ additional bits, for some computable function $h$. Finally, since every contracted $\forall$ computation tree of $M$ accepting $(x,k)$ within prescribed resources embeds into $U_{N,K}$, it is straightforward to see from the construction that $M'$ accepts $(x,k)$ within the prescribed resources if and only if $M$ does. This concludes the proof of Lemma~\ref{lem:regularization}.

\subsubsection{XSLP membership for regular AROSM}
\label{sec:regular-csp-arXiv}
What remains to prove Theorem \ref{thm:machine-characterization} is the proof of Lemma~\ref{lem:regular-reduction}, which we give next.

Let $M$ be a regular AROSM that decides a parameterized problem $Q$ so that if an input $(x,k)$ with $n=|x|$ is accepted by $M$, there is an accepting $\forall$ computation tree whose contraction is $T_{n,k}$, and moreover $M$ uses $f(k)\log n$ nondeterminism, $f(k)+a\log n$ conondeterminism, $f(k)$ alternation, $f(k)\log n$ stack size, and $f(k)+a\log n$ working space, for a constant $a$ and a computable function $f$. Our goal is, given an input $(x,k)$, to construct an instance of $\csp$ that is satisfiable if and only if $(x,k)$ is accepted by $M$. The construction should work in parameterized logspace. This will be clear from the description, hence we leave the verification of this computability claim to the reader.

\newcommand{\Tb}{S}

For convenience, denote $A=\lfloor f(k)\log n\rfloor$, $B=f(k)+\lfloor a\log n\rfloor$, and $K=f(k)$. We first construct the tree $T_{n,k}$. We write $T=T_{n,k}$ for brevity and we may assume that the depth of $T$ is at most $K$.
Then, from $T$ we obtain a tree $\Tb$ by subdividing every edge of $T$ $3K$ times, that is, replacing every edge with a path of length $3K$.
Thus, every vertex of $S$ is either {\em{principal}} --- it originates from $T$ --- or {\em{subdivision}}. For an edge $uv$ of $T$, where $u$ is the parent of $v$, the subdivision vertices created for $uv$ are called, in order from $u$,
$$s_{uv}^1,s_{uv}^2,\ldots,s_{uv}^K,x_{uv}^1,x_{uv}^2,\ldots,x_{uv}^{2K}.$$
Note that the depth of $\Tb$ is at most $3K^2+K$.
Intuitively, the subdivision vertices will accommodate information about the computation within conondeterministic and nondeterministic blocks, which in particular includes the bits pushed onto the stack in those blocks.

We now define a $\csp$ instance $I=(G,\{D(u)\colon u\in V(G)\},\{C(u,v)\colon uv\in E(G)\},\Tb)$ as follows. The variable set of $I$ is equal to the vertex set of $\Tb$: $V(G)=V(\Tb)$. In the construction we will maintain the assertion that whenever we add a constraint between variables $x$ and $y$, $x$ and $y$ will be in the ancestor/descendant relation in $\Tb$. Thus, $\Tb$ will be an elimination tree of the Gaifman graph $G$ of depth at most $3K^2+K$.

We now proceed to defining the domains of the variables.
Let $C$ be the set of stackless configurations of $M$; note that $|C|\leq h(k)\cdot n^{b}$ for some $b\in \N$ and computable $h$. Further, let $P$ be the set of all binary strings of length at most $\log n$, and let
$$D=C\times P\times [0,A]\times [0,B]\times [0,A].$$
We now define the domains as follows.
\begin{itemize}
\item If $r$ is the root vertex of $\Tb$, we set
$$D(r)=\{(c_0,\varepsilon,0,0,0)\},$$
where $c_0$ is the starting configuration of $M$ and $\varepsilon$ is the empty string.
\item For every non-root vertex $u$ of $S$, we set
$$D(u)=D,$$
with two exceptions.
First, if $u$ is a leaf, then we restrict $D(u)$ to only those tuples $(c,\cdot,\cdot,\cdot,\cdot)\in D$ in which $c$ is a final configuration.
Second, if $u$ is a principal vertex, then we restrict $D(u)$ to only those tuples of the form $(c,\cdot,\cdot,\cdot,\cdot)$ in which the number of leaves of $U(c)$ matches the number of children of $u$ in $S$. Here, $U(c)$ is the universal block rooted at $c$, which can be computed by Claim~\ref{cl:generate-block}.
\end{itemize}
The reader should read the consecutive entries of tuples used in the domains as follows: for $(c,s,n_\exists,n_\forall,d)\in D$,
\begin{itemize}
    \item $c$ is the current stackless configuration;
    \item $s$ is a string consisting of at most $\log n$ symbols most recently pushed onto the stack;
    \item $n_\exists,n_\forall$ are the bounds on the number of nondeterministic and conondeterministic bits used so far;
    \item $d$ is the current total length of the stack.
\end{itemize}

Next, we define the constraints, which come in three different types: {\em{conondeterministic block constraints}}, {\em{nondeterministic block constraints}}, and {\em{acceptance constraints}}.

We start with conondeterministic block constraints.
Consider an edge $uv$ of $T$, where $u$ is the parent of $v$. Let $i$ be such that $v$ is the $i$th child of $u$. For a moment, fix a value $q=(c,s,n_\exists,n_\forall,d)\in D(u)$. Let $\ell$ be the $i$th leaf of $U(c)$, the universal block rooted at $c$. We define the {\em{$(q,i)$-guided sequence}} as the sequence $$(q^1,q^2,\ldots,q^K)\in D(s_{uv}^1)\times \ldots \times D(s_{uv}^K)$$ defined as follows: $q^j=(c^j,s^j,n^j_\exists,n^j_\forall,d^j)$, where:
\begin{itemize}
    \item $c^j$ is the last configuration on the branch of $U(c)$ from $c$ to $\ell$ where the total number of symbols pushed onto the stack during the transitions from $c$ to $c^j$ does not exceed $j\lceil \log n\rceil$. If the total number of symbols pushed onto the stack during transitions from $c$ to $\ell$ is smaller than $j\lceil \log n\rceil$, we set $c^j=\ell$.
    \item $s^j$ is the word pushed onto the stack during the transitions from $c^{j-1}$ to $c^j$ along the considered branch, where $c^0=c$ by convention.
    \item $n^j_\exists=n_\exists$.
    \item $n^j_\forall=n_\forall+b_j$, where $b_j$ is the total number of universal configurations between $c$ and $c_j$ in $U(c)$.
    \item $d^j=d^{j-1}+|s^j|$, where $d^0=d$ by convention.
\end{itemize}
Note that using Claim~\ref{cl:generate-block} and straightforward emulation of $M$, for every $q\in D(u)$ we can compute the $(q,i)$-guided sequence in parameterized logspace.
Hence, for each $j\in \{1,\ldots,K\}$ we may add a constraint $$C(u,s^j_{uv})=\{(q,q^j)\colon q^j\textrm{ is the }j\textrm{th entry of the }(q,i)\textrm{-guided sequence}\}.$$
These constitute the conondeterminstic block constraints.

Next, we define the nondeterministic block constraints. Again, these will be defined for every edge $uv$ of $T$, where $u$ is the parent of $v$. We first define relation $E\subseteq D\times D$ as follows: for $q=(c,s,n_\exists,n_\forall,d)$ and $q'=(c',s',n'_\exists,n'_\forall,d')$, we add $(q,q')$ to $E$ if and only if there is a sequence of transition of $M$ satisfying the following conditions:
\begin{itemize}
    \item The sequence transitions from $c$ to $c'$.
    \item Along the sequence, at most $\log n$ push operations take place, and the word pushed onto the stack by those operations is $s'$. Further, $d'=d+|s'|$. 
    \item Along the sequence, at most $\log n$ existential configurations are encountered. Further, $n_\exists'=n_\exists+a$, where $a$ is the number of those configurations.
    \item Except possibly for the last configuration $c'$, no configuration along the sequence is universal. Further, $n_\forall'=n_\forall$.
\end{itemize}
Observe that $E$ can be computed in parameterized logspace by a straightforward emulation of $M$ for every possible sequence of at most $\log n$ nondeterministic bits. Then, for each $j\in \{1,\ldots,2K,2K+1\}$ we introduce a constraint
$$C(x^{j-1}_{uv},x^{j}_{uv})=E,$$
where by convention, $x^{0}_{uv}=s^K_{uv}$ and $x^{2K+1}_{uv}=v$. These constitute the nondeterministic block constraints.

Finally, we are left with acceptance constraints. Define a relation $F\subseteq D\times D$ as follows. Consider $q=(c,s,n_\exists,n_\forall,d)$ and $q'=(c',s',n'_\exists,n'_\forall,d')$, and let $p$ be the integer whose binary encoding is the content of the working tape of $c'$. Then 
we add $(q,q')$ to $F$ if one of the following two conditions hold:
\begin{itemize}
    \item $p\notin \{d-|s|+1,d-|s|+2,\ldots,d\}$, or
    \item $p\in \{d-|s|+1,d-|s|+2,\ldots,d\}$ and the $(p-d+|s|)$th symbol of $s$ is $1$.
\end{itemize}
Then for every leaf $\ell$ of $\Tb$ and every ancestor $w$ of $\ell$ in $\Tb$, we add the constraint
$$C(w,\ell)=F.$$
These constitute the acceptance constraints.

This finishes the construction of the instance $I$. It is now straightforward to see that $I$ is satisfiable if and only if $M$ has a $\forall$ computation tree on $(x,k)$ whose contraction is $T$ and which uses 
at most $f(k)\log n$ nondeterminism, $f(k)+a\log n$ conondeterminism, $f(k)+a\log n$ working space, and $f(k)\log n$ stack size. Since $M$ is regular, this is equivalent to the assertion that $M$ accepts $(x,k)$. This finishes the proof of Lemma~\ref{lem:regular-reduction}. The proof of Theorem~\ref{thm:machine-characterization} is also complete.

\newcommand{\Af}{\mathbb{A}}
\newcommand{\ar}{\mathrm{ar}}
\newcommand{\pr}{\mathsf{parent}}
\newcommand{\rt}{\mathsf{root}}
\newcommand{\forb}{\mathsf{forbidden}}
\newcommand{\forest}{\mathsf{forest}}
\newcommand{\dm}{\mathsf{domain}}

\subsection{A logic characterization}\label{sec:logic-char}
We now provide another characterization of $\XSLP$, by providing a complete problem related to model-checking first-order logic. This reflects the definitions of classes $\mathrm{AW}[\star]$ and of the $\mathrm{A}$-hierarchy, see~\cite[Chapter~8]{FlumGrohe}.

We use the standard terminology for relational structures. A (relational) signature is a set $\Sigma$ consisting of {\em{relation symbols}}, where each relation symbol $R$ has a prescribed arity $\mathrm{ar}(R)\in \N$. A {\em{$\Sigma$-structure}} $\Af$ consists of a universe $U$ and, for every relation symbol $R\in \Sigma$, its {\em{interpretation}} $R^{\Af}\subseteq U^{\ar(R)}$ in $\Af$. In this paper we only consider {\em{binary signatures}}, that is, signatures where every relation has arity at most $2$.

For a signature $\Sigma$, we may consider standard first-order logic over $\Sigma$-structures. In this logic there are variables for the elements of the universe. Atomic formulas are of the form $x=y$ and $R(x_1,\ldots,x_k)$ for some $R\in \Sigma$ with $k=\ar(R)$, with the obvious semantics. These can be used to form larger formulas by using Boolean connectives, negation, and quantification (both existential and universal).

A $\Sigma$-structure $\Af$ is called {\em{forest-shaped}} if $\Sigma$ contains a binary relation $\pr$ such that $\pr^\Af$ is the parent relation on a rooted forest with the vertex set being the universe of $\Af$, and a unary relation $\rt$ such that $\rt^\Af$ is the set of roots of this forest. We say that a first-order sentence $\varphi$ over $\Sigma$ is {\em{$\forall$-guided}} if it is of the form:
\begin{eqnarray*}
\varphi & = &\forall x_1\, \exists y_1\, \ldots\, \forall x_k\, \exists y_k\,\left(\rt(x_1)\wedge \pr(x_1,x_2)\wedge \ldots \wedge \pr(x_{k-1},x_k)\right)\Rightarrow\\
& & \qquad\qquad\qquad\qquad\qquad \psi(x_1,y_1,\ldots,x_k,y_k)
\end{eqnarray*}
where $\psi$ is quantifier-free. In other words, $\varphi$ is in a prenex form starting with a universal quantifier, and moreover we require that the first universally quantified variable is a root and every next universally quantified variable is a child of the previous one. Note that there are no restrictions on existential quantification.

For a binary signature $\Sigma$, we consider the problem of model-checking $\forall$-guided formulas on forest-shaped $\Sigma$-structures. In this problem we are given a forest-shaped $\Sigma$-structure $\Af$ and a $\forall$-guided sentence $\varphi$, and the question is whether $\varphi$ holds in $\Af$. We consider this as a parameterized problem where $\|\varphi\|$ --- the total length of an encoding of the sentence $\varphi$ --- is the parameter. 

The following statement provides a characterization of $\XSLP$ in terms of the model-checking problem described above.
\begin{restatable}{theorem}{logic}\label{thm:logic-characterization}
There exists a binary signature $\Sigma$ such that the following conditions are equivalent for a parameterized problem $Q$.
\begin{enumerate}[(1)]
\item\label{a:xslp2} $Q\in \XSLP$;
\item\label{a:fo} $Q$ can be $\pl$-reduced to the problem of model-checking $\forall$-guided sentences on forest-shaped $\Sigma$-structures.
\end{enumerate}
\end{restatable}
\begin{proof}
We first prove \eqref{a:xslp2} $\Rightarrow$ \eqref{a:fo}. By the definition of $\XSLP$, it suffices to show a $\pl$-reduction from $\sfrac{\csp}{\td}$ to the problem of model-checking $\forall$-guided sentences on forest-shaped $\Sigma$-structures, for some binary signature $\Sigma$. We will use a signature $\Sigma$ consisting of three binary relations and two unary relations:
$$\Sigma=\{\pr(\cdot,\cdot),
\forb(\cdot,\cdot),\dm(\cdot,\cdot),\rt(\cdot),\forest(\cdot)\}.$$
Our task is, given an instance $I=(G,\{D(u)\colon u\in V(G)\},\{C(u,v)\colon uv\in E(G)\},T)$ of $\sfrac{\csp}{\td}$, to construct an equivalent instance $(\Af,\varphi)$ of model-checking.

The structure $\Af$ is defined as follows:
\begin{itemize}
\item The universe of $\Af$ is $V(G)\cup \bigcup_{u\in V(G)} D(u)$, where without loss of generality we assume that the domains $\{D(u)\colon u\in V(G)\}$ are pairwise disjoint.
\item Relation $\forest^\Af(\cdot)$ selects the vertices of $V(T)$, relation $\pr^\Af(\cdot,\cdot)$ is the parent/child relation in $T$, and relation $\rt^\Af(\cdot)$ selects the roots of $T$ and the elements of $\bigcup_{u\in V(G)} D(u)$ (so formally, the forest present in structure $\Af$ is $T$ plus every element of $\bigcup_{u\in V(G)} D(u)$ as a separate root).
\item Relation $\dm^\Af(\cdot,\cdot)$ binds every vertex of $G$ with its domain, that is,
$$\dm^{\Af}=\{(u,a)\colon u\in V(G), a\in D(u)\}.$$
\item Relation $\forb^\Af(\cdot,\cdot)$ selects all pairs of values that are {\em{forbidden}} by the constraints of~$I$. That is,
$$
\forb^{\Af}=\bigcup_{uv\in E(G)} (D(u)\times D(v))\setminus C(u,v).
$$
\end{itemize}
It remains to define $\varphi$. We set 
\begin{eqnarray}
\label{eq:mink}
\varphi & = &\forall x_1\, \exists y_1\, \ldots\, \forall x_k\, \exists y_k\ \left(\rt(x_1)\wedge \pr(x_1,x_2)\wedge \ldots \wedge \pr(x_{k-1},x_k)\right)\Rightarrow\nonumber\\
& & \qquad\qquad\qquad\qquad\qquad \psi(x_1,y_1,\ldots,x_k,y_k),
\end{eqnarray}
and
$$\psi(x_1,y_1,\ldots,x_k,y_k)=\forest(x_1)\wedge \bigwedge_{i=1}^k \dm(x_i,y_i)\wedge \bigwedge_{1\leq i<j\leq k} \neg \forb(y_i,y_j).$$
It is straightforward to verify that $\varphi$ is satisfied in $\Af$ if and only if the input instance $I$ is satisfiable.

\bigskip

Next, we prove \eqref{a:fo} $\Rightarrow$ \eqref{a:xslp2}. By Theorem~\ref{thm:machine-characterization}, it suffices to solve, for any fixed binary signature~$\Sigma$, the model-checking problem for $\forall$-guided formulas on forest-shaped $\Sigma$-structures using an AROSM within the resource bounds stipulated in Theorem~\ref{thm:machine-characterization}. For this, let $\Af$ be the input $\Sigma$-structure and let $\varphi$ be the sentence whose satisfaction in $\Af$ is in question. We may assume that $\varphi$ is of the form~\eqref{eq:mink} and additionally $\psi$ is in DNF. Note that $k\leq \|\varphi\|$ and $\|\psi\|\leq \|\varphi\|$. Let $F$ be the rooted forest defined by the relation $\pr$ in $\Af$, and let $F^\circ$ be the rooted tree obtained from $F$ by (i) removing all vertices at depth larger than $k$, and (ii) adding a new root $r$ and making all the roots of $F$ children of $r$. Additionally, we may apply Lemma~\ref{lem:conon} to compute a labeling $\lambda$ of edges of $F^\circ$ such that for $\lambda$ is injective on edges connecting any node of $F^{\circ}$ with its children, and the total length of labels on every branch of $F^{\circ}$ is at most $\lceil\log (n+1)\rceil$, where $n$ is the size of the universe of $\Af$.

The constructed AROSM $M$ works as follows on input $(\Af,\varphi)$. $M$ proceeds in $k$ round. At the beginning of round $i$, $M$ has on its working tape a node $u_{i-1}$ of $F^{\circ}$ that is the evaluation of the previously quantified universal variable $x_{i-1}$; initially we set $u_0=r$. Then round $i$ works as follows:
\begin{itemize}
\item Conondeterministically guess the evaluation $u_i$ of the next universal variable $x_i$ by guessing the label $\lambda(u_{i-1}u_i)$. This can be done using $2|\lambda(u_{i-1}u_i)|+1$ conondeterministic bits as explained in the footnote on page~\pageref{pg:ft}. Push $u_i$ onto the stack.
\item Nondeterministically guess the evaluation $v_i$ of the next existential variable $y_i$ using $\lceil\log n\rceil$ nondeterministic bits. Push $v_i$ onto the stack.
\end{itemize}
Here, we assume that the identifiers of elements of $\Af$ are length-$\lceil \log n\rceil$ binary strings. Note that thus, $M$ has so far used at most $k\lceil \log n\rceil$ nondeterministic bits and at most $k+2\lceil \log n\rceil$ conondeterministic bits, and alternation $2k$.

Once all the $k$ rounds are done, the stack contains elements $u_1,v_1,\ldots,u_k,v_k$ and it remains to verify that $\psi(u_1,v_1,\ldots,u_k,v_k)$ holds. To do this, $M$ performs the following operations (recall that $\psi$ is in DNF):
\begin{itemize}
\item Nondeterministically guess the disjunct of $\psi$ that is satisfied.
\item Conondeterministically guess an atomic formula within the disjunct guessed above.
\item If this atomic formula is of the form $R(z,z')$ for some $z,z'\in \{x_1,y_1,\ldots,x_k,y_k\}$ and binary relation $R$, then do the following. First, nondeterministically guess two elements $w,w'$  of $\Af$ such that $R(w,w')$ holds. Then, verify that $w$ and $w'$ are the evaluations of $z$ and $z'$ by conondeterministically guessing a bit of the encoding of $(w,w')$ to be checked, and checking it against the corresponding bit on the stack.
\item If the atomic formula is the equality of two variables or a check of relation of arity at most $1$, perform an analogous check.
\end{itemize}
Observe that these operations increase the alternation by $4$, and use in addition $\Oh(\log \|\psi\|\log n)$ bits of nondeterminism and $\Oh(\log \|\psi\|+\log n)$ bits of conondeterminism. Hence, machine $M$ works within the stipulated resources and it is clear that $M$ accepts $(\Af,\varphi)$ if and only if $\varphi$ holds in $\Af$.
\end{proof}

\section{The \texorpdfstring{$\S$}{S}-hierarchy and \texorpdfstring{$d$}{d}-fold vertex cover}
\label{sec:st_hierarchy}
In this section, we prove Proposition~\ref{prop:S_inclusions}. 
\shierarchy*
We write $\vc_d(G)$ for the $d$-fold vertex cover number of $G$ and $\td(G)$ for the treedepth of $G$. As mentioned in the introduction, for a graph $G$ we have $\vc_d(G)\leq k$ if and only if there exists a rooted tree $T$ of depth at most $d$ and a partition $\{V_t\colon t\in V(T)\}$ of the vertex set of $G$ such that $|V_t|\leq k$ for every $t\in V(T)$, and whenever an edge of $G$ connects a vertex of $V_s$ with a vertex of $V_t$, $s$ and $t$ must be in the ancestor/descendant relation in $T$. We will call such a structure a {\em{$k$-fat elimination tree}} of~$G$. Let us observe the following.

\begin{observation}
\label{obs:d_foldvc_relations}
For a graph $G$ and integer $d\geq 1$, we have $\td(G)\leq d \cdot  \vc_d(G)$. Moreover, $\vc_d(G)$ is  at most the minimum size of a modulator to a graph of treedepth at most $d-1$.
\end{observation}

Before we proceed, let us argue that if $\vc_d(G)\leq k$, then a witnessing $k$-fat elimination tree of depth at most $d$ can be computed in time fpt in $k$ and $d$. This is an easy application of Courcelle's theorem.

\begin{lemma}\label{lem:fat-compute}
There exists an algorithm that given a graph $G$ on $n$ vertices and integers $k$ and $d$, runs in time $f(k,d)\cdot n$ for a computable function $f$, and either decides that $\vc^d(G)>k$ or finds a $k$-fat elimination tree of $G$ of depth at most $d$.
\end{lemma}
\begin{proof}
We assume the reader's familiarity with logic $\mathsf{MSO}_2$ and Courcelle's theorem, see~\cite[Section~7.4]{CyganFKLMPPS15} for an introduction.

First note that there is an $\mathsf{MSO}_2$ sentence $\psi_{d,k}$ stating that a graph has $d$-fold vertex cover number at most $k$; such a statement can be easily written inductively directly from the definition. Using $\psi_{d-1,k}$, we can write a formula $\varphi_{d,k}(X)$ such that $G\models \varphi_{d,k}(A)$ if and only if $|A|\leq k$ and every connected component of $G-A$ has $(d-1)$-fold vertex cover number at most $k$. Now, we can compute a $k$-fat elimination forest of $G$ of depth at most $d$, or conclude its nonexistence, as follows:
\begin{itemize}
\item Apply Courcelle's theorem to $G$ and $\varphi_{d,k}$ to find $A$ such that every connected component of $G-A$ 
has the $(d-1)$-fold vertex cover number at most $k$. If there is no such $A$, then we can output that $\vc_d(G)>k$.
\item Apply the algorithm recursively to every connected component $C$ of $G-A$ with depth bound $d-1$, thus obtaining a $k$-fat elimination tree $F^C$ of depth at most $d-1$. Then a $k$-fat elimination tree of $G$ of depth at most $d$ can be obtained by taking the disjoint union of trees $F^C$, adding $A$ as the new root, and making all former roots into children of $A$.
\end{itemize}
By Observation~\ref{obs:d_foldvc_relations}, the treewidth of $G$ is at most $dk$, or otherwise we can conclude that $\vc_d(G)>k$; so the application of Courcelle's theorem is justified. As each application of Courcelle's theorem takes linear fpt time, and the recursion depth is bounded by $d$, it follows that the algorithm runs in time $f(d,k)\cdot n$ for a computable function $f$.
\end{proof}

We are ready to prove Proposition~\ref{prop:S_inclusions}.

\begin{proof}[Proof of Proposition~\ref{prop:S_inclusions}]
Recall that, by definition,  $\csp$ parameterized by the $d$-fold vertex cover number is $\S[2d-1]$-complete and $\csp$ parameterized by treedepth is in $\XSLP^+$. Further, we proved in Theorem \ref{thm:oddW} that $\csp$ parameterized by the size of a modulator to a graph of treedepth at most $d-1$ is $\W[2d-1]$-hard, and $\csp$ parameterized by the vertex count of the graph (aka $1$-fold vertex cover number) is $\W[1]$-hard, as it can easily model \textsc{Multicolor Clique}. Since $\S[2d-1]$ and $\XSLP^+$ are closed under fpt-reductions, Observation \ref{obs:d_foldvc_relations}
implies that for all integers $d\geq 1$,
    $\W[2d-1]\subseteq \S[2d-1]\subseteq \XSLP^+$. 

We are left with the inclusion $\S[2d-1]\subseteq \A[2d-1]$.
It suffices to show that $\csp$ parameterized by $d$-fold vertex cover can be reduced to the problem of model-checking a given first-order $\Sigma_{2d-1}$-sentence in a given relational structure over some fixed signature $\Sigma$, because the latter problem is, by definition, complete for $\A[2d-1]$ (see  \cite[Definition~5.7]{FlumGrohe}). We repeat the definition of $\Sigma_{t}$ for convenience of the reader below.
We write $\Sigma_0,\Pi_0$ for the class of quantifier-free formulas, and for $t \geq  0$, we
let $\Sigma_{t+1}$ be the set of all formulas of the form
$\exists x_1\dots \exists x_k~\phi$
where $\phi \in \Pi_t$, and let $\Pi_{t+1}$ be the set of all formulas of the form
$\forall x_1\dots \forall x_k ~ \phi$ where $\phi \in \Sigma_t$. In other words, formulas in $\Sigma_{2d-1}$ start with an existential quantifier and have at most $2d-1$ blocks of same quantifiers. 

Consider an instance $I=(G,\{D(v)\colon v\in V(G)\},\{C(u,v)\colon uv\in E(G)\})$ of $\csp$ parameterized by $d$-fold vertex cover number of $G$, which we denote by $k$. By Lemma~\ref{lem:fat-compute}, in fpt time we can compute a $k$-fat elimination tree $T$ of $G$, where every node $t\in V(T)$ is associated with a vertex subset $V_t$ of size at most $k$. By adding dummy variables if necessary, we may assume that $|V_t|=k$ for all $t\in V(T)$ and every leaf of $T$ is at depth $d$.

\newcommand{\tr}{\mathsf{tree}}
\newcommand{\bag}{\mathsf{bag}}

We now construct a relational structure $\Af$ similarly as in the proof of Theorem~\ref{thm:logic-characterization}. Structure $\Af$ is over a signature $\Sigma$ that consists of four binary and one unary relation:
$$\Sigma=\{\rt(\cdot),\pr(\cdot,\cdot),
\bag(\cdot,\cdot),\dm(\cdot,\cdot),\forb(\cdot,\cdot)\}.$$
The universe of $\Af$ is the disjoint union of $V(T)$, $V(G)$, and $\bigcup_{u\in V(G)} D(u)$, where we assume without loss of generality that the domains are pairwise disjoint. Relation $\rt^\Af$ selects only the root of $T$, while relation $\pr^\Af$ is the parent/child relation in $T$. Next, relation $\bag^\Af$ encodes the partition $\{V_t\colon t\in V(T)\}$ as follows:
$$\bag^\Af=\{(t,u)\colon t\in V(T), u\in V_t\}.$$
Also, relation $\dm^\Af$ binds domains with respective vertices:
$$\dm^\Af=\{(u,a)\colon u\in V(G), a\in D(u)\}.$$
Finally, relation $\forb^\Af(\cdot,\cdot)$ selects all pairs of values that are {\em{forbidden}} by the constraints, that is,
$$
\forb^{\Af}=\bigcup_{uv\in E(G)} (D(u)\times D(v))\setminus C(u,v).
$$
This concludes the construction of $\Af$.

We are left with constructing a $\Sigma_{2d-1}$-sentence $\varphi$ whose satisfaction in $\Af$ is equivalent to the satisfiability of $I$. We write it as follows:
\begin{eqnarray*}
\varphi = \exists x_1~ \exists y_1^1\exists z_1^1\ldots \exists y_1^k\exists z_1^k\ 
\forall x_2~
\exists y_2^1\exists z_2^1\ldots \exists y_2^k\exists z_2^k~
\forall x_3~\ldots \\
 \ldots\
\forall x_d ~\exists y_d^1\exists z_d^1\ldots \exists y_d^k\exists z_d^k\ \rt(x_1)\wedge
\left(\bigwedge_{j=2}^d \pr(x_{j-1},x_j)\Rightarrow \psi\right),
\end{eqnarray*}
where
\begin{eqnarray*}
\psi & = & \bigwedge_{j=1}^d\, \bigwedge_{i=1}^k \bag(x_j,y_j^i)\wedge \bigwedge_{j=1}^d\ \bigwedge_{1\leq i<i'\leq k} y_j^i\neq y_j^{i'}\wedge \\
& & \bigwedge_{j=1}^d\ \bigwedge_{i=1}^k \dm(y_j^i,z_j^i) \wedge \bigwedge_{(i,j)\in [k]\times [d]}\ \bigwedge_{(i',j')\in [k]\times [d]} \neg \forb(z_i^j,z_{i'}^{j'}).
\end{eqnarray*}
Clearly, $\varphi$ is a $\Sigma_{2d-1}$-formula.
Further, it is straightforward to see that $\Af\models \varphi$ if and only if $I$ is satisfiable.
This proves the desired membership in $\A[2d-1]$.
\end{proof}

\section{Discussion of Table \ref{table:complexityoverview}, Precoloring Extension and List Coloring}
\label{section:table_listcol}
In this section we discuss the results reported in Table~\ref{table:complexityoverview}. In particular, we show how all the claimed  findings follow, either directly or as straightforward corollaries, from the results of this paper or from the literature.

As mentioned, a classic observation of  Freuder~\cite{Freuder90} is that $\csp$ can be solved in $Sn^{k+\Oh(1)})$ on graphs
of treewidth $k$.
A seminal result by Marx~\cite{Marx10} showed that we cannot expect to replace treewidth here by another graph
parameter that does not imply bounded treewidth. (For the precise statement, see \cite{Marx10}.)
It was long known that $\csp$ with the treewidth of the Gaifman graph as the parameter is $\W[1]$-hard (see~\cite{GottlobSS02}). However, there were works on additional  additional parameterizations that render the problem fixed-parameter tractable~\cite{GottlobSS02,SamerS10}. Papadimitriou and Yannakakis~\cite{PapadimitriouY99} showed that 
$\csp$ is $\W[1]$-complete when the number of variables (vertices in the graph) is the parameter.

Fellows et al.~\cite{FellowsFLRSST11} showed that \textsc{List Coloring} is $\W[1]$-hard for graphs of bounded treewidth.
This was improved by 
Fiala et al.~\cite{FialaGK11}, who showed that \textsc{List Coloring} parameterized by the vertex
cover number is $\W[1]$-hard. 

\textsc{List Coloring} with the size of
a vertex or edge modulator to a \emph{clique} was studied in \cite{BanikJPR20,GolovachPS14,GutinMOW21}.
As for $\csp$, one can observe directly that $\csp$ is NP-complete for cliques.

A well-known $\W[1]$-complete problem is \textsc{Multicolor Clique}, introduced in ~\cite{FellowsHRV09}: we are given a graph $G$, and a coloring of the vertices
with $k$ colors, and ask whether one can choose one vertex from each color so that the chosen vertices form a clique (of size $k$); $k$ is the parameter. We can turn an instance of \textsc{Multicolor Clique} into an instance of $\csp$ with $k$ vertices as follows: take for each color in $G$ a vertex in $H$, for each vertex in $G$ a color that the corresponding vertex in $H$ can choose, and for each edge in $G$ a pair of forbidden vertices in $H$. This simple transformation with its reverse, give
also a proof that $\csp$ is $\W[1]$-complete with the number of vertices as parameter. 
Generalizations of \textsc{Multicolor Clique} were given in \cite{BodlaenderGNS21,BodlaenderGJJPP22}. A problem called
\textsc{Chained Multicolor Clique} was shown to be $\XNLP$-complete~\cite{BodlaenderGNS21}, and a problem called \textsc{Tree-Chained
Multicolor Clique} was shown to be $\XALP$-complete~\cite{BodlaenderGJJPP22}. Using the same transformation as above, we
obtain instances of $\csp$ with parameterizations by respectively pathwidth plus degree, or bandwidth (these are $\XNLP$-complete);
and by treewidth plus degree, or tree-partition width plus degree (these are $\XALP$-complete). Tree-partition width is an old graph parameter; it was introduced under the name strong treewidth
by Seese in 1985~\cite{Seese85}, and gives a useful tree structure for graph algorithmic studies (e.g.,~\cite{BodlaenderCW22,Wood09}.)

One easily observes that \textsc{List Coloring}
with the number of vertices $n$ as the parameter has a polynomial kernel (and thus is in FPT): just remove every vertex whose list has
size at least $n$ (as they can always be colored). If the maximum degree of a graph together with its treewidth, pathwidth or tree-partition width, then one can remove
all vertices whose list size is larger than their degree, and solve the remaining instance in linear time with dynamic programming.

The new membership results for \textsc{List Coloring} are proved in Section~\ref{sec:listcol_membership} and we give the remaining claimed hardness results below. Also, we provide some discussion of the related {\sc{Precoloring Extension}} problem.

\subsection{List Coloring: new hardness results}
\label{section:listcoloring}
We deduce hardness results for \textsc{List Coloring} from the hardness results for $\csp$ proved in this paper.
\begin{corollary}
    \label{corollary:lcmod}
Let $d\geq 2$. Then
    \textsc{List Coloring} is $\W[2d-1]$-hard
    with
    the size of a modulator to a treedepth-$d$ graph as the parameter.
\end{corollary}

\begin{proof}
We use a
well-known and easy reduction from $\csp$ to \textsc{List Coloring}. 
Take an instance $I=(G,\{D(u)\colon u\in V(G)\},\{C(u,v)\colon uv\in E(G)\})$ of $\csp$. For each edge $uv\in E(G)$, remove this edge, and
instead add for each ``forbidden'' pair $(c,c') \in (D(u) \times D(v)) \setminus C(u,v)$ a new vertex with an edge
to $u$ and an edge to $v$, and give this new vertex the color set $\{c,c'\}$. One easily sees that this
gives an instance of \textsc{List Coloring} that is equivalent to $I$ in terms of satisfiability.

The operation increases the treedepth of a graph by at most one: use the same tree for the original vertices.
For a new vertex $x$ with neighbours $u$ and $v$, we have that $v$ is an ancestor of $w$ or vice versa. Let $x$
be a child of the lower of these two vertices.

Thus we have a parameterized reduction from \textsc{List Coloring} with parameter a modulator to treedepth $d$
to $\csp$ with parameter a modulator to treedepth $d+1$, and the result follows
from Lemma~\ref{lemma:w2d+1hard}.
\end{proof}

The case $d=2$ also gives an interesting corollary, by
noting that a graph with treedepth~$2$ is a forest
consisting of isolated vertices and stars (graphs of the
form $K_{1,r}$).

\begin{corollary}
    \textsc{List Coloring} is $\W[3]$-hard with as parameter the size of a vertex modulator to forest of depth two, and thus for 
    parameterization by feedback vertex set size. 
\end{corollary}

Similarly, we obtain the following result from the XSLP-completeness of $\csp$ parameterized by treedepth.
\begin{corollary}
    \textsc{List Coloring} with treedepth as parameter is $\XSLP$-complete.
\end{corollary}

\subsection{Precoloring Extension}
\label{section:precoloring extension}
The \textsc{Precoloring Extension} problem is a further
special case of $\csp$. We are given a graph $G$, a set of colors
$\mathcal{C}$, a subset $W\subseteq V$, a partial coloring $f\colon W \rightarrow \mathcal{C}$, and ask if we can extend $f$ to
a proper coloring of $G$, i.e., is there an $f'\colon V(G) \rightarrow \mathcal{C}$, such that for every $w\in W$, we have $f'(w)=f(w)$,
and for every edge $xy\in E(G)$, we have $f(x)\neq f(y)$.
We deduce the following result.
\begin{corollary}
    Let $d\geq 2$.
    \textsc{Precoloring Extension} is fixed-parameter tractable with the vertex cover as the parameter, and $\W[2d-3]$-hard and in $\W[2d-2]$ with the size of
    a modulator to treedepth-$d$ as the parameter.
\end{corollary}
\begin{proof}
We first show that \textsc{Precoloring Extension} is fixed-parameter tractable
with the vertex cover number as the parameter. Suppose we have an instance $(G,\mathcal{C},W,f)$ of \textsc{Precoloring Extension}, together with a vertex cover $S\subseteq V(G)$ of $G$. Write $I=V(G)\setminus S$ and $k=|S|$.

If $|\mathcal{C}|\leq k$, then dynamic programming can be used to solve the problem~\cite{JansenS97}. So we may assume that there are at least $k+1$ colors, which means that each vertex in $I\setminus W$  can always be colored.
We build an equivalent \textsc{List Coloring} instance: we give each vertex $v\in S\setminus W$ the color list obtained from $\mathcal{C}$ by removing all colors assigned by $f$ to neighbors of $v$ in $W$, removing all vertices in $I\cup W$ and removing all vertices in $S\setminus W$ whose list is of size at least $k$. 
We have an equivalent instance of \textsc{List Coloring} on at most $k$ vertices, with lists of size at most $k$, so we have an instance that can be described with $\Oh(k^2\log k)$ bits. This is a polynomial kernel, which then can be solved by any brute-force algorithm for \textsc{List Coloring}.

A similar observation can be used to transform \textsc{Precoloring Extension} on graphs with a modulator of size $k$ to a
treedepth-$d$ graphs to instances of \textsc{List Coloring} on graphs with a modulator of size $k$ to
a treedepth $d-1$ graph. Suppose we have an instance $(G,\mathcal{C},W,f)$
with a vertex set $S$ such that $G-S$ has treedepth at most $d$.
Let $T$ be an elimination forest of $G- S$ of depth at most $d$. 

We distinguish the cases that $|\mathcal{C}|\leq d+k$, and $|\mathcal{C}|> d+k$. The treewidth of a graph with a
 modulator of size $k$ to a treedepth-$d$ graph is at most $d+k$, and thus, in the former case, the problem can be solved in FPT time when parameterized by $d+k$~\cite{JansenS97}.

Now assume $|\mathcal{C}|> d+k$.
Consider a leaf $v$ in $T$. If $v$ is not precolored, i.e., $v\not\in W$, then $v$ can be removed as it has degree at most $d+k$ (its ancestors in $T$ and the vertices in $W$ can be neighbors). We transform to a \textsc{List Coloring} instance, by repeatedly removing leafs that are not precolored, then removing all precolored vertices and giving the remaining vertices a list of colors obtained from $\C$ by removing the colors of its precolored neighbours.
This gives an equivalent instance, and since all leaves are precolored when we remove the precolored vertices, the depth of $T$ has been decreased by at least one. Hence $\W[2d-2]$ membership follows from Corollary~\ref{cor:listcol_even_membership}.

To deduce the hardness result, we start with an instance of \textsc{List Coloring} with a graph~$G$, lists $L(v)$ for all $v\in V(G)$, a modulator
$W$ of size $k$, and an elimination forest $T$ of $G-W$ of depth $d-1$.

For each $v\in V(G)$ and color $c\in \C\setminus L(v)$, we add a vertex $a_{v,c}$ adjacent to $v$ precolored by $c$. We extend $T$ to an elimination forest $T'$ of depth at most $d$, where for all $v\in V(G)\setminus W$, the vertices of the form $a_{v,c}$ are added as children of $v$, and for all $w\in W$, the vertices of the form $a_{w,c}$ are added as rooted trees consisting of a single vertex. We then forget the lists and have obtained an equivalent instance. The $\W[2d-3]$-hardness result now follows from Corollary~\ref{corollary:lcmod}.
\end{proof}

\section{Conclusion}
\label{section:conclusions}
In this paper we explored the parameterized complexity of $\csp$ for a variety of relatively strong structural parameters, including the vertex cover number, treedepth, and several modulator-based parameters. We believe that together with the previous works on $\XALP$ and $\XNLP$~\cite{BodlaenderCW22,BodlaenderGJJL22,BodlaenderGJJPP22,BodlaenderGNS21,ElberfeldST15}, our work uncovers a rich complexity structure within the class~$\XP$, which is worth further exploration. We selected concrete open questions below.
\begin{itemize}
    \item In~\cite{BodlaenderGJJPP22,BodlaenderGNS21}, several problems such as \textsc{Independent Set} or \textsc{Dominating Set}, which are fixed-parameter tractable when parameterized by treewidth, were shown to be $\XALP$- and $\XNLP$-complete when parameterized by the {\em{logarithmic}} treewidth and pathwidth, which is at most $k$ when the corresponding width measure is at most $k\log n$. Can one prove similar results for the class $\XSLP$ and parameterization by logarithmic treedepth?
    \item Theorem~\ref{thm:oddW} provides natural complete problems only for the odd levels of the $\W$-hierarchy. Similarly, we defined the $\S$-hierarchy only for odd levels. It would be interesting to have a natural description of the situation also for the even levels.
    \item The characterizations of $\XSLP$ given by Theorems~\ref{thm:machine-characterization} and~\ref{thm:logic-characterization} can be ``projected'' to a rough characterizations of classes $\S[d]$ for odd $d$ by stipulating that the alternation is at most $d$. Unfortunately, this projection turns out not to be completely faithful: the obtained problems do not precisely characterize the class $\S[d]$, but lie somewhere between $\S[d-\Oh(1)]$ and $\S[d+\Oh(1)]$. Can we provide a compelling description of the levels of the $\S$-hierarchy in terms of machine problems or in terms of model-checking first-order logic?
    \item What is the complexity of \textsc{List Coloring} parameterized by the vertex cover number? Currently, we know it is $\W[1]$-hard and in $\W[2]$. Similarly, what is the complexity of \textsc{List Coloring} and \textsc{Precoloring Extension} with the minimum size of a modulator to a treedepth-$d$ graph as the parameter?
    \item  Can one obtain a better understanding of the complexity of $\csp$ and \textsc{List Coloring} parameterized by the feedback vertex number?
\end{itemize}

\bibliography{references}

\end{document}